\documentclass[12pt,journal,onecolumn]{IEEEtran}

\usepackage{amsmath}
\usepackage{amssymb}
\usepackage{amsthm}
\usepackage{graphicx}

\newtheorem{defi} {Definition}
\newtheorem{lemma} {Lemma}
\newtheorem{prop} {Proposition}
\newtheorem{cor} {Corollary}
\newtheorem*{rem} {Remark}

\begin{document}
\title{General Spectrum Sensing in Cognitive Radio Networks}

\author{Sheng-Yuan~Tu~and
        Kwang-Cheng~Chen, \IEEEmembership{Fellow,~IEEE}
%
\thanks{This work was supported in part by the National Science Council of
Taiwan, under the contract NSC97-2219-E-002-018.}
\thanks{Sheng-Yuan Tu is with the
Institute of Communication Engineering, National Taiwan University,
Taiwan (R.O.C) (e-mail: shinetu@santos.ee.ntu.edu.tw).}
\thanks{Kwang-Cheng Chen is with the
Institute of Communication Engineering and the Department of
Electrical Engineering, National Taiwan University, Taiwan (R.O.C)
(e-mail: chenkc@cc.ee.ntu.edu.tw).}}



\maketitle

\begin{abstract}
It is well known that the successful operation of cognitive radio
(CR) between CR transmitter and CR receiver (CR link) relies on
reliable spectrum sensing. To network CRs requires more information
from spectrum sensing beyond traditional techniques, executing at CR
transmitter and further information regarding the spectrum
availability at CR receiver. Redefining the spectrum sensing along
with statistical inference suitable for cognitive radio networks
(CRN), we mathematically derive conditions to allow CR transmitter
forwarding packets to CR receiver under guaranteed outage
probability, and prove that the correlation of localized spectrum
availability between a cooperative node and CR receiver determines
effectiveness of the cooperative scheme. Applying our novel
mathematical model to potential hidden terminals in CRN, we
illustrate that the allowable transmission region of a CR, defined
as neighborhood, is no longer circular shape even in a pure path
loss channel model. This results in asymmetric CR links to make
bidirectional links generally inappropriate in CRN, though this
challenge can be alleviated with the aid of cooperative sensing.
Therefore, spectrum sensing capability determines CRN topology. For
multiple cooperative nodes, to fully utilize spectrum availability,
the selection methodology of cooperative nodes is developed due to
limited overhead of information exchange. Defining reliability as
information of spectrum availability at CR receiver provided by a
cooperative node and by applying neighborhood area, we can compare
sensing capability of cooperative nodes from both link and network
perspectives. In addition, due to dynamic network topology lack of
centralized coordination in CRN, CRs can only acquire local and
partial information in limited sensing duration, robust spectrum
sensing is therefore proposed to ensure successful CRN operation.
Limits of cooperative schemes and their impacts on network operation
are also derived.
\end{abstract}


\IEEEpeerreviewmaketitle

\section{Introduction}
\IEEEPARstart{C}{ognitive} radios (CR)
\cite{GSSinCRN:Mitola00}\cite{GSSinCRN:Haykin05}, having capable of
sensing spectrum availability, is considered as a promising
technique to alleviate spectrum scarcity due to current static
spectrum allotment policy \cite{GSSinCRN:FCC}. Traditional CR link
availability is solely determined by the spectrum sensing conducted
at the transmitter (i.e. CR-Tx). If the CR-Tx with packets to relay
senses the selected channel to be available, it precedes this
opportunistic transmission. To facilitate the spectrum sensing, at
time instant $t_n$, we usually use a hypothesis testing as follows.
\begin{equation} \label{e27}
    \begin{aligned}
    H_1: Y&=I+N\\
    H_0: Y&=S+I+N
    \end{aligned}
\end{equation}
where $Y$ means the observation at CR-Tx; $S$ represents signal from
primary system (PS); $I$ is the interference from co-existing
multi-radio wireless networks; $N$ is additive white Gaussian noise
(AWGN). They are all random variables at time $t_n$. We can conduct
this hypothesis testing in several ways based on different
criterions and different assumptions
\cite{GSSinCRN:Ma09}-\cite{GSSinCRN:Cabric04}:

\begin{enumerate}
    \item{Energy Detection \cite{GSSinCRN:Digham03}-\cite{GSSinCRN:Penna09}:} Energy detection is widely considered
    due to its simple complexity and no need of \emph{a priori} knowledge
    of PS. However, due
    to noise and interference power uncertainty, the performance of
    energy detection severely degrades \cite{GSSinCRN:Sahai04}\cite{GSSinCRN:Zeng07}, and the detector
    fails to differentiate PS from the interference.
    \item{Cyclostationary Detection \cite{GSSinCRN:Tu09}-\cite{GSSinCRN:Guo09}:}
    Stationary PS signal can be exploited to achieve a better and
    more robust detector. Stationary observation coupled with the periodicity
    of carriers, pulse trains, repeating spreading code, etc., results in
    a cyclostationary model for the signal, which can be exploited in
    frequency domain by the spectral correlation function \cite{GSSinCRN:Gardner88}\cite{GSSinCRN:Gardner91},
    with high computational complexity and long observation duration.
    \item{Locking Detection \cite{GSSinCRN:Cabric04}\cite{GSSinCRN:Cabric06}:} In practical communication systems, pilots
    and preambles are usually periodically transmitted to facilitate
    synchronization and channel estimation, etc. These known signals
    can be used for locking detection to distinguish PS from
    noise and the interference. However, locking detection requires
    more \emph{a priori} information about PS, including frame
    structure, modulation types, and coding schemes,
    etc.
    \item{Covariance-based Detection \cite{GSSinCRN:Zeng07}-\cite{GSSinCRN:Zayen09}:} Because of the dispersive
    channels, the utility of multiple antennas, or even
    over-sampling, the signals from PS are correlated and can be
    utilized to differentiate PS from white noise. The existence of
    PS can be determined by evaluating the structure of covariance
    matrix of the received signals. The detector can be implemented
    blindly \cite{GSSinCRN:Zayen09} by singular value decomposition, that is, it requires
    no \emph{a priori} knowledge of PS and noise, but needs good computational complexity.
    \item{Wavelet-based Detection \cite{GSSinCRN:Tian06}:} The detection is implemented
    by modeling the entire wideband
    signals as consecutive frequency subbands where the power
    spectral characteristic is smooth within each subband but
    exhibits an abrupt change between neighboring subbands, at the price
    of high sampling rate due to wideband signals.
\end{enumerate}

However, above spectrum sensing mechanisms, focusing on physical
layer detection or estimation at CR-Tx, ignore the spectrum
availability at CR receiver. We could illustrate the insufficiency
of traditional spectrum sensing model, especially to network CRs.
Due to existence of fading channels and noise uncertainty along with
limited sensing duration \cite{GSSinCRN:Ghasemi082}, even when there
is no detectable transmission of PS during this venerable period,
the receiver of this opportunistic transmission (i.e. CR-Rx) may
still suffer from collisions from simultaneous transmission(s), as
Fig. \ref{Fig_1} shows. The CR-Rx locates in the middle of CR-Tx and
PS-Tx and PS activities are hidden to CR-Tx, which induces a
challenge to spectrum sensing. We can either develop more powerful
sensing techniques such as cooperative sensing
\cite{GSSinCRN:Unnikrishnan08}-\cite{GSSinCRN:Mishra06} to alleviate
hidden terminal problem, or a more realistic mathematical model for
spectrum sensing what we are going to do hereafter.

The organization of this paper is as follows. We elaborate a
realistic definition of link availability and system model in
Section II and present general spectrum sensing with/without
cooperation in Section III an IV respectively. In Section V,
realistic operation of CRN is suggested to investigate impacts of
spectrum sensing and cooperative scheme on network operation.
Numerical results and examples are illustrated in Section VI.
Finally, conclusions are made in Section VII.

\section{System Model}
From a viewpoint of information theory, spectrum sensing can be
modeled as a binary channel that transmit CR link availability (one
bit information in link layer) to CR-Tx with transition
probabilities representing spectrum sensing capability, probability
of missing detection and probability of false alarm. Therefore,
traditional spectrum sensing mechanisms could be explained by a
mathematical structure of defining link availability.
\begin{defi}
CR link availability, between CR-Tx and CR-Rx, is specified by an
indicator function
\begin{equation*}
\mathbf{1}^{link}=
    \begin{cases}
    1, &\text{CR link is available for opportunistic transmission}\\
    0, &\text{otherwise}
    \end{cases}
\end{equation*}
\end{defi}
\begin{defi} \label{d1}
CR-Tx senses the spectrum and determines link availability based on
its observation as
\begin{equation*}
\mathbf{1}^{Tx}=
    \begin{cases}
    1, &\text{CR link is available for transmission at CR-Tx}\\
    0, &\text{otherwise}
    \end{cases}
\end{equation*}
\end{defi}
\begin{lemma} \label{l5}
Traditional spectrum sensing for CR link suggests
$\mathbf{1}^{link}=\mathbf{1}^{Tx}$.
\end{lemma}

The definition of link availability is pretty much similar to the
clear channel assessment (CCA) of medium access control (MAC) in the
IEEE 802.11 wireless local area networks (WLAN), or the medium
availability indicator in \cite{GSSinCRN:Chen09}. We may have a
correspondence between link availability in dynamic channel access
of cognitive radio networks (CRN), and the CCA in MAC of WLAN.

As we explain in Figure \ref{Fig_1} and/or take interference into
testing scenario, we may note that \textbf{Lemma \ref{l5}} is not
generally true. To generally model spectrum sensing, including
hidden terminal scenarios, we have to reach two simultaneous
conditions: (1) CR-Tx senses the link available to transmit (2)
CR-Rx can successfully receive packets, which means no PS signal at
CR-Rx side, nor significant interference to prohibit successful CR
packet reception (i.e. beyond a target SINR). In other words, at
CR-Rx,
\begin{equation}  \label{e28}
    \text{SINR}_{CR-Rx}=\frac{P_{CR-Tx}}{P_{PS}+P_I+P_N}\geq
    \eta_{outage}
\end{equation}
where $\eta_{outage}$ is the SINR threshold at CR-Rx for outage in
reception over fading channels, $P_{CR-Tx}$ is the received power
from opportunistic transmission from CR-Tx, $P_{PS}$ is the power
contributed from PS simultaneous operation for general network
topology such as ad hoc, $P_I$ is the total interference power from
other co-existing radio systems \cite{GSSinCRN:Ghasemi08}, and $P_N$
is band-limited noise power, with assuming independence among PS,
CR, interference systems, and noise.

Based on this observation, CR link availability should be composed
of localized spectrum availability at CR-Tx and CR-Rx, which may not
be identical in general. The inconsistency of spectrum availability
at CR-Tx and CR-Rx is rarely noted in current literatures. However,
this factor not only suggests spatial behavior for CR-Tx and CR-Rx
but also is critical to some networking performance such as
throughput of CRN, etc. \cite{GSSinCRN:Jafar07} developed a
brilliant two-switch model to capture distributed and dynamic
spectrum availability. However, \cite{GSSinCRN:Jafar07} focused on
capacity from information theory and it is hard to directly extend
the model in studying network operation of CRN. Actually, two
switching functions can be generalized as indicator functions to
indicate the activities of PS based on the sensing by CR-Tx and
CR-Rx respectively \cite{GSSinCRN:Chen09}. Generalizing the concept
of \cite{GSSinCRN:Jafar07}\cite{GSSinCRN:Srinivasa07} to facilitate
our study in spectrum sensing and further impacts on network
operation, we represent the spectrum availability at CR-Rx by an
another indicator function.
\begin{defi}
The true availability for CR-Rx can be indicated by
\begin{equation*}
\mathbf{1}^{Rx}=
    \begin{cases}
    1, &\text{CR link is available for reception at CR-Rx}\\
    0, &\text{otherwise}
    \end{cases}
\end{equation*}
\end{defi}
Please note that the activity of PS estimated at CR-Rx in
\cite{GSSinCRN:Jafar07} may not be identical to $1^{Rx}$. That is,
even when CR-Rx senses that PS is active, CR-Rx may still
successfully receive packets from CR-Tx if the received power from
CR-Tx is strong enough to satisfy (\ref{e28}). We call this
rate-distance nature \cite{GSSinCRN:Chen07} that is extended from an
overlay concept \cite{GSSinCRN:Srinivasa07}. Therefore, we consider
a more realistic mathematical model for CR link availability that
can be represented as multiplication (i.e. AND operation) of the
indicator functions of spectrum availability at CR-Tx and CR-Rx to
satisfy two simultaneous conditions for CR link availability.
\begin{prop} \label{p1}
$\mathbf{1}^{link}=\mathbf{1}^{Tx}\mathbf{1}^{Rx}$
\end{prop}

To obtain the spectrum availability at CR-Rx (i.e.
$\mathbf{1}^{Rx}$) and to eliminate hidden terminal problem, a
handshake mechanism has been proposed by sending Request To Send
(RTS) and Clear To Send (CTS) frames. However, the effectiveness of
RTS/CTS degrades in general ad hoc networks \cite{GSSinCRN:Xu02}.
Furthermore, since CRs have lower priority in the co-existing
primary/secondary communication model, CRs should cherish the
venerable duration for transmission and reduce the overhead caused
by information exchange and increases spectrum utilization
accordingly. Therefore, the next challenge would be that
$\mathbf{1}^{Rx}$ cannot be known \emph{a priori} at CR-Tx, due to
no centralized coordination nor information exchange in advance
among CRs when CR-Tx wants to transmit. As a result, general
spectrum sensing turns out to be a composite hypothesis testing. In
this paper, we introduce statistical inference that is seldom
applied in traditional spectrum sensing to predict/estimate spectrum
availability at CR-Rx and to regard it as performance lower bound in
general spectrum sensing.

Further examining \textbf{Proposition~\ref{p1}}, we see that
prediction of $\mathbf{1}^{Rx}$ is necessary when
$\mathbf{1}^{Tx}=1$, which is equivalent to prediction of
$\mathbf{1}^{link}$. In this paper, we model $\mathbf{1}^{Rx}$ when
$\mathbf{1}^{Tx}=1$ as a Bernoulli process with the probability of
spectrum availability at CR-Rx
$\Pr(\mathbf{1}^{Rx}=1|\mathbf{1}^{Tx}=1)=\alpha$. The value of
$\alpha$ exhibits spatial behavior of CR-Tx and CR-Rx and thus
impacts of hidden terminal problem. If $\alpha$ is large, CR-Rx is
expected to be close to CR-Tx and hidden terminal problem rarely
occurs (and vise versa).

\section{General Spectrum Sensing}
The prediction of $\mathbf{1}^{Rx}$ at CR-Tx can be modeled as a
hypothesis testing, that is, detecting $\mathbf{1}^{Rx}$ with
\emph{a priori} probability $\alpha$ but no observation. To design
optimum detection, we consider minimum Bayesian risk criterion,
where Bayesian risk is defined by
\begin{equation} \label{e30}
    R=w\Pr(\mathbf{1}^{link}=0|\mathbf{1}^{Tx}=1)P_F+\Pr(\mathbf{1}^{link}=1|\mathbf{1}^{Tx}=1)P_M
\end{equation}
In (\ref{e30}),
$P_F=\Pr(\hat{\mathbf{1}}^{link}=1|\mathbf{1}^{link}=0,\mathbf{1}^{Tx}=1)$,
$P_M=\Pr(\hat{\mathbf{1}}^{link}=0|\mathbf{1}^{link}=1,\mathbf{1}^{Tx}=1)$,
and $w\geq 0$ denotes the normalized weighting factor to evaluate
costs of $P_F$ and $P_M$, where $\hat{\mathbf{1}}^{link}$ represents
prediction of $\mathbf{1}^{link}$. We will show that the value of
$w$ relates to the outage probability of CR link in Section V.

Since $\mathbf{1}^{Rx}$ is unavailable at CR-Tx, we have to develop
techniques to "obtain" some information of spectrum availability at
CR-Rx. Inspired by the CRN tomography
\cite{GSSinCRN:Yu09}\cite{GSSinCRN:Yu092}, we may want to derive the
statistical inference of $\mathbf{1}^{Rx}$ based on earlier
observation. It is reasonable to assume that CR-Tx can learn the
status of $\mathbf{1}^{Rx}$ at previous times when
$\mathbf{1}^{Tx}=1$, which is indexed by $n$. That is, at time $n$,
CR-Tx can learn the value of
$\mathbf{1}^{Rx}[n-1],\mathbf{1}^{Rx}[n-2],\ldots$. In other words,
we can statistically infer $\mathbf{1}^{Rx}[n]$ from
$\mathbf{1}^{Rx}[n-1]$, $\mathbf{1}^{Rx}[n-2]$, $\ldots$,
$\mathbf{1}^{Rx}[n-L]$, where $L$ is the observation depth. This
leads to a classical problem from Bayesian inference.
\begin{lemma} \label{l8}
Through the Laplace formula \cite{GSSinCRN:Casella02}, the estimated
probability of spectrum availability at CR-Rx is
\begin{equation} \label{e1}
    \hat{\alpha}=\frac{N+1}{L+2}
\end{equation}
where $N=\sum_{l=1}^{L}{\mathbf{1}^{Rx}[n-l]}$.
\end{lemma}

\begin{prop}
Inference-based spectrum sensing at CR-Tx thus becomes
\begin{equation} \label{e13}
    \hat{\mathbf{1}}^{link}=
    \begin{cases}
    \mathbf{1}^{Tx}, &\text{if $\hat{\alpha}\geq\frac{w}{w+1}$}\\
    0, &\text{otherwise}
    \end{cases}
\end{equation}
where $\hat{\alpha}$ is in (\ref{e1}).
\end{prop}
\begin{proof}
Since the optimum detector under Bayesian criterion is the
likelihood ratio test \cite{GSSinCRN:Poor94}, we have
\begin{equation} \label{e41}
    \frac{\Pr(\mathbf{1}^{Rx}=1|\mathbf{1}^{Tx}=1)}{\Pr(\mathbf{1}^{Rx}=0|\mathbf{1}^{Tx}=1)}=
    \frac{\alpha}{1-\alpha}
    {\substack{
    \overset{\hat{\mathbf{1}}^{link}=1}{\geq}\\
    \underset{\hat{\mathbf{1}}^{link}=0}{<}
    }}
    \frac{C_{01}-C_{00}}{C_{10}-C_{11}}=w
\end{equation}
where $C_{ij}$ denotes the cost incurred by determining
$\hat{\mathbf{1}}^{link}=j$ when $\mathbf{1}^{link}=i$. According to
Bayesian risk in (\ref{e30}), we have $C_{11}=C_{00}=0$ and
$C_{01}/C_{10}=w$. Rearranging the inequality, we obtain the
proposition.
\end{proof}
\begin{rem}
CR-Tx believes CR link is available and forwards packets to CR-Rx if
the probability of spectrum available at CR-Rx $\alpha$ is high
enough. Otherwise, CR-Tx is prohibited from using the link even when
CR-Tx feels free for transmission because it can generate
unaffordable cost, that is, intolerable interference to PS or
collisions at CR-Rx.
\end{rem}

\section{General Cooperative Spectrum Sensing}
\subsection{Single Cooperative Node}
Spectrum sensing at cooperative node, which can be represented
$\mathbf{1}^{Co}$, is to explore more information about
$\mathbf{1}^{Rx}$ and therefore alleviates hidden terminal problem.
We can use Fig. \ref{Fig_2} to depict the scenario. In case the
existence of obstacles, $\mathbf{1}^{Rx}$ is totally orthogonal to
$\mathbf{1}^{Tx}$. $\mathbf{1}^{Co}$ is useful simply because of
more correlation between $\mathbf{1}^{Co}$ and $\mathbf{1}^{Rx}$.
From above observation, we only care about correlation of
$\mathbf{1}^{Rx}$ and $\mathbf{1}^{Co}$ when $\mathbf{1}^{Tx}=1$ and
assume
\begin{align*}
    \Pr(\mathbf{1}^{Co}=1|\mathbf{1}^{Rx}=1,\mathbf{1}^{Tx}=1)&=\beta\\
    \Pr(\mathbf{1}^{Co}=0|\mathbf{1}^{Rx}=0,\mathbf{1}^{Tx}=1)&=\gamma
\end{align*}
Thus the correlation between $\mathbf{1}^{Co}$ and
$\mathbf{1}^{Rx}$, $\rho$, and corresponding properties become
\begin{equation}
    \rho =\frac{\sqrt{\alpha(1-\alpha)}(\beta+\gamma-1)}
    {\sqrt{(\alpha\beta+(1-\alpha)(1-\gamma))
    (\alpha(1-\beta)+(1-\alpha)\gamma)}}
\end{equation}
\begin{lemma} \label{l1}
$\rho$ is a strictly concave function with respect to
$\alpha\in(0,1)$ if $1<\beta+\gamma<2$ but a strictly convex
function if $0<\beta+\gamma<1$. In addition, $\mathbf{1}^{Co}$ and
$\mathbf{1}^{Rx}$ are independent if and only if $\rho=0$, i.e,
$\beta+\gamma=1$.
\end{lemma}
\begin{proof}
Let
\begin{equation*}
    f(\alpha)=\sqrt{\frac{\alpha(1-\alpha)}{(\alpha\beta+(1-\alpha)(1-\gamma))
    (\alpha(1-\beta)+(1-\alpha)\gamma)}}
\end{equation*}
Then taking first order and second order differentiation with
respect to $\alpha$, we have
\begin{align*}
    f'(\alpha)&=\frac{-\alpha^2\beta(1-\beta)+(1-\alpha)^2\gamma(1-\gamma)}
    {2(\alpha(1-\alpha))^{1/2}(\alpha\beta+(1-\alpha)(1-\gamma))^{3/2}
    (\alpha(1-\beta)+(1-\alpha)\gamma)^{3/2}}\\
    f{''}(\alpha)&=K(\alpha)[-10\alpha^2(1-\alpha)^2\beta\gamma(1-\beta)(1-\gamma)-
    4\alpha^4(1-\alpha)\beta(1-\beta)(\beta\gamma+(1-\beta)(1-\gamma))\\
    &+(3-4\alpha)\alpha^4\beta^2(1-\beta)^2-
    4\alpha(1-\alpha)^4\gamma(1-\gamma)(\beta\gamma+(1-\beta)(1-\gamma))\\
    &+(-1+4\alpha)(1-\alpha)^4\gamma^2(1-\gamma)^2]
\end{align*}
where $K(\alpha)>0$ for $\alpha,\beta,\gamma\in(0,1)$ and
$\beta+\gamma\neq1$. In addition,
$\beta\gamma+(1-\beta)(1-\gamma)-\beta(1-\beta)=
    \beta(\beta+\gamma-1)+(1-\beta)(1-\gamma)
    =\beta\gamma+(1-\beta)(1-\beta-\gamma)>0$
for $\beta,\gamma\in(0,1)$. Similarly, we have
$\beta\gamma+(1-\beta)(1-\gamma)>\gamma(1-\gamma)$ for
$\beta,\gamma\in(0,1)$. Therefore, combining the second and the
third terms, and the fourth and the last terms in the bracket of
$f{''}(\alpha)$, we have
\begin{align*}
    f{''}(\alpha)&<K(\alpha)[-10\alpha^2(1-\alpha)^2\beta\gamma(1-\beta)(1-\gamma)
    -\alpha^4\beta^2(1-\beta)^2-(1-\alpha)^4\gamma^2(1-\gamma)^2]
    <0
\end{align*}
if $\alpha\in(0,1)$. Since $\rho=(\beta+\gamma-1)f(\alpha)$, we
prove the first statement of the lemma. For the second statement,
obviously, $\rho=0$ if $\mathbf{1}^{Co}$ and $\mathbf{1}^{Rx}$ are
independent. Reversely, if $\rho=0$, i.e., $\beta+\gamma=1$, we have
\begin{align*}
    \Pr(\mathbf{1}^{Co}=1|\mathbf{1}^{Tx}=1)&=
    \sum_{s=0}^{1}{\Pr(\mathbf{1}^{Rx}=s|\mathbf{1}^{Tx}=1)\Pr(\mathbf{1}^{Co}=1|\mathbf{1}^{Rx}=s,\mathbf{1}^{Tx}=1)}\\
    &=(1-\alpha)(1-\gamma)+\alpha\beta=(1-\alpha)\beta+\alpha\beta=\beta\\
    &=\Pr(\mathbf{1}^{Co}=1|\mathbf{1}^{Rx}=1,\mathbf{1}^{Tx}=1)\\
    &=1-\gamma=\Pr(\mathbf{1}^{Co}=1|\mathbf{1}^{Rx}=0,\mathbf{1}^{Tx}=1)
\end{align*}
Similarly, we can show
$\Pr(\mathbf{1}^{Co}=0|\mathbf{1}^{Tx}=1)=\Pr(\mathbf{1}^{Co}=0|\mathbf{1}^{Rx}=1,\mathbf{1}^{Tx}=1)
=\Pr(\mathbf{1}^{Co}=0|\mathbf{1}^{Rx}=0,\mathbf{1}^{Tx}=1)$ and
complete the proof.
\end{proof}

By statistical inference, CR-Tx can learn statistical characteristic
of $\mathbf{1}^{Rx}$ and $\mathbf{1}^{Co}$, i.e.,
$\{\alpha,\beta,\gamma\}$, by previous observations. From a
viewpoint of hypothesis testing, we would like to detect
$\mathbf{1}^{Rx}$ with \emph{a priori} probability $\alpha$ and one
observation $\mathbf{1}^{Co}$, which is the detection result at the
cooperative node. In addition, probability of detection and
probability of false alarm at the cooperative node are $\beta$ and
$1-\gamma$ respectively.
\begin{prop} \label{p4}
Spectrum sensing with one cooperative node becomes
\begin{equation} \label{e15}
    \hat{\mathbf{1}}^{link}=
    \begin{cases}
    \mathbf{1}^{Tx},  &\text{if $\alpha\geq\max\{\alpha_1,\alpha_2\}$}\\
    \mathbf{1}^{Tx}\mathbf{1}^{Co}, &\text{if $\alpha_2<\alpha
    <\alpha_1, \rho>0$}\\
    \mathbf{1}^{Tx}\bar{\mathbf{1}}^{Co}, &\text{if $\alpha_1<\alpha
    <\alpha_2, \rho<0$}\\
    0,   &\text{if $\alpha\leq\min\{\alpha_1,\alpha_2\}$}
    \end{cases}
\end{equation}
where $\bar{\mathbf{1}}^{Co}$ is the complement of
$\mathbf{1}^{Co}$, $\alpha_1=w\gamma/(1-\beta+w\gamma)$ and
$\alpha_2=w(1-\gamma)/(\beta+w(1-\gamma))$.
\end{prop}
\begin{proof}
The likelihood ratio test based on observed signal $\mathbf{1}^{Co}$
can be written as follows. For $\mathbf{1}^{Co}=1$,
\begin{equation*}
    \frac{\Pr(\mathbf{1}^{Co}=1|\mathbf{1}^{Rx}=1,\mathbf{1}^{Tx}=1)}
    {\Pr(\mathbf{1}^{Co}=1|\mathbf{1}^{Rx}=0,\mathbf{1}^{Tx}=1)}=
    \frac{\beta}{1-\gamma}
    {\substack{
    \overset{\hat{\mathbf{1}}^{link}=1}{\geq}\\
    \underset{\hat{\mathbf{1}}^{link}=0}{<}
    }}
    \frac{w(1-\alpha)}{\alpha}
\end{equation*}
For $\mathbf{1}^{Co}=0$,
\begin{equation*}
    \frac{\Pr(\mathbf{1}^{Co}=0|\mathbf{1}^{Rx}=1,\mathbf{1}^{Tx}=1)}
    {\Pr(\mathbf{1}^{Co}=0|\mathbf{1}^{Rx}=0,\mathbf{1}^{Tx}=1)}=
    \frac{1-\beta}{\gamma}
    {\substack{
    \overset{\hat{\mathbf{1}}^{link}=1}{\geq}\\
    \underset{\hat{\mathbf{1}}^{link}=0}{<}
    }}
    \frac{w(1-\alpha)}{\alpha}
\end{equation*}
Rearranging the above inequalities, we obtain the proposition.
\end{proof}

It is interesting to note that cooperative spectrum sensing is not
always helpful, that is, it does not always further decrease
Bayesian risk. We see that if $\alpha$ is large (greater than
$\max\{\alpha_1,\alpha_2\}$), that is, hidden terminal problem
rarely occurs because either CR-Rx is close to CR-Tx or CR-Tx adopt
cooperative sensing to determine $\mathbf{1}^{Tx}$, prediction of
$\mathbf{1}^{Rx}$ is unnecessary at CR-Tx. On the other hand, if
$\alpha$ is small (less than $\min\{\alpha_1,\alpha_2\}$), CR-Tx is
prohibited from forwarding packets to CR-Rx even with the aid of
cooperative sensing.

In the following, we adopt minimum error probability criterion
(i.e., $w=1$) and give an insight into the condition that
cooperative sensing is helpful. Although we set $w=1$, we do not
lose generality because we can scale \emph{a priori} probability
$\alpha$ to $\alpha/(\alpha+w(1-\alpha))$ as $w\neq1$. Applying
\textbf{Lemma \ref{l1}} and the fact that
$\rho|_{\alpha=\alpha^1_C}=\rho|_{\alpha=\alpha^2_C}$ when $w=1$, we
can reach the following corollary.
\begin{cor} \label{cor1}
If we adopt minimum error probability criterion, spectrum sensing
with one cooperative node becomes
\begin{equation} \label{e16}
    \hat{\mathbf{1}}^{link}=
    \begin{cases}
    \mathbf{1}_{[\alpha\geq1/2]}\mathbf{1}^{Tx},  &\text{if $|\rho|\leq\Psi$}\\
    (\mathbf{1}_{[\rho>0]}\mathbf{1}^{Co}+\mathbf{1}_{[\rho<0]}\bar{\mathbf{1}}^{Co})\mathbf{1}^{Tx}, &\text{if $|\rho|>\Psi$}
    \end{cases}
\end{equation}
where $\mathbf{1}_{[s]}$ is an indicator function, which is equal to
1 if the statement $s$ is true else equal to 0, and
\begin{equation*}
    \Psi=\left|\frac{\beta+\gamma-1}{\sqrt{2(\beta\gamma+(1-\beta)(1-\gamma))}}\right|
\end{equation*}
\end{cor}

\begin{rem}
The effectiveness of a cooperative node only depends on the
correlation of spectrum availability at CR-Rx and the cooperative
node. If the correlation is low, information provided by the
cooperative node is irrelevant to the spectrum sensing which
degenerates to (\ref{e13}).
\end{rem}

By establishing a simple indicator model in link layer, we
mathematically demonstrate the limit of a cooperative node in
general spectrum sensing. It is natural to ask what will happen for
multiple cooperative nodes and how to compare sensing capability
among cooperative nodes. In the following, we provide metrics to
measure sensing capability of cooperative nodes from link and
network perspectives.

\subsection{Preliminaries} Before exploring multiple cooperative nodes, we
introduce notations and properties to systematically construct
relation between joint probability mass function (pmf) and marginal
pmf of spectrum availability at cooperative nodes. We first define
notations in the following.
\begin{defi}
For an $m\times n$ matrix $\mathbf{A}$ and two $n\times 1$ vectors
$\mathbf{u}$ and $\mathbf{v}$,
\begin{equation*}
    \begin{matrix}
    \mathbf{A}[i,\star]:&\text{the } i\text{th row of } \mathbf{A} &
    \mathbf{A}[\star,j]:&\text{the } j\text{th column of } \mathbf{A}\\
    \mathbf{A}^T:&\mathbf{A}^T[i,j]=\mathbf{A}[j,i] &
    \mathbf{A}^{RC}:&\mathbf{A}^{RC}[\star,j]=\mathbf{A}[\star,n-j+1]\\
    \mathbf{A}^{RR}:&\mathbf{A}^{RR}[i,\star]=\mathbf{A}[m-i+1,\star] &
    \mathbf{1}_{m\times n}:&\mathbf{A}[i,j]=1\\
    \mathbf{0}_{m\times n}:&\mathbf{A}[i,j]=0 &
    \mathbf{I}_{n}:&n\times n \text{ identity matrix}\\
    \mathbf{u}\odot\mathbf{v}:&\mathbf{u}\odot\mathbf{v}[i]=\mathbf{u}[i]\mathbf{v}[i]&
    \mathbf{u}\preceq\mathbf{v}:&\mathbf{u}[i]\leq \mathbf{v}[i]\\
    \|\mathbf{u}\|_p:&(\sum_i{|\mathbf{u}[i]|^p})^{1/p} &
    \mathbf{1}_{[\mathbf{u}\geq0]}:&\mathbf{1}_{[\mathbf{u}\geq0]}[i]=\mathbf{1}_{[\mathbf{u}[i]\geq0]}
    \end{matrix}
\end{equation*}
\end{defi}

Let
\begin{align}
    \mathbf{A}^{n}_k&=
    \begin{bmatrix}
    \mathbf{A}^n_{0,k} & \mathbf{A}^n_{1,k} & \cdots & \mathbf{A}^n_{k,k}
    \end{bmatrix} \quad 0\leq n\leq k \label{e3} \\
    \mathbf{G}^{(1)}_{m,k}&=
    \begin{bmatrix}
    (\mathbf{A}^0_k)^T & (\mathbf{A}^1_k)^T & (\mathbf{A}^2_k)^T & \cdots & (\mathbf{A}^m_k)^T
    \end{bmatrix}^T \label{e20} \\
    \mathbf{G}^{(0)}_{m,k}&=(\mathbf{G}^{(1)}_{m,k})^{RC} \label{e4}
\end{align}
where $\mathbf{A}^n_{m,k}$ is a $\binom{k}{n}\times\binom{k}{m}$
matrix, $\mathbf{A}^0_{m,k}=\mathbf{1}_{1\times\binom{k}{m}},0\leq
m\leq k$, $\mathbf{A}^k_{m,k}=\mathbf{0}_{1\times\binom{k}{m}},
0\leq m\leq k-1$, $\mathbf{A}^k_{k,k}=1$, and for $1\leq n\leq k-1$
\begin{align*}
    \mathbf{A}^n_{0,k}&=\mathbf{0}_{\binom{k}{n}\times1} \qquad
    \mathbf{A}^n_{k,k}=\mathbf{1}_{\binom{k}{n}\times 1} \\
    \mathbf{A}^n_{m,k}&=
    \begin{bmatrix}
    \mathbf{A}^n_{m,k-1} & \mathbf{A}^n_{m-1,k-1} \\
    \mathbf{0}_{\binom{k-1}{n-1}\times\binom{k-1}{m}} & \mathbf{A}^{n-1}_{m-1,k-1}
    \end{bmatrix},
    \quad 1\leq m\leq k-1
\end{align*}

The role of $\mathbf{A}^{1}_k$ (or $\mathbf{A}^1_{m,k},0\leq m\leq
k$) is to specify arrangements of joint pmf and marginal pmf such
that their relation can be easily established by
$\mathbf{G}^{(s)}_{m,k}, s=0,1$. In the following, we show
properties of $\mathbf{A}^n_{m,k}$ and $\mathbf{G}^{(s)}_{m,k}$.

\begin{lemma} \label{l2}
Let $\mathcal{I}_{m,k}(j)=\{i|\mathbf{A}^1_{m,k}[i,j]=1,1\leq i\leq
k\}$, $1\leq j\leq\binom{k}{m}$ and we have
\begin{align} \label{e17}
    |\mathcal{I}_{m,k}(j)|&=m\\
    \mathcal{I}_{m,k}(j)&=\mathcal{I}_{m,k}(l)\quad \text{if and only if
    $j=l$} \label{e22} \\
    \mathbf{A}^n_{m,k}[i,j]&=\mathbf{1}_{[\mathcal{I}_{m,k}(j)\supseteq\mathcal{I}_{n,k}(i)]}
    \label{e18}
\end{align}
where $|\mathcal{I}|$ denotes number of elements in the set
$\mathcal{I}$.
\end{lemma}

\begin{rem}
Since $\mathbf{A}^1_{m,k}$ is a $k\times\binom{k}{m}$ matrix, from
(\ref{e17}), (\ref{e22}), we conclude that
$\mathbf{A}^1_{m,k}[\star,j],1\leq j\leq\binom{k}{m}$ contains all
possible permutations of $m$ ones and $k-m$ zeros.
\end{rem}

\begin{lemma} \label{t1}
Let $S_m=\sum_{i=0}^{m}{\binom{k}{i}}$. Let
$\overline{\mathbf{G}}^{(1)}_{m,k}$ and
$\underline{\mathbf{G}}^{(0)}_{m,k}$ be $S_m\times S_m$ matrices,
$\underline{\mathbf{G}}^{(1)}_{m,k}$ and
$\overline{\mathbf{G}}^{(0)}_{m,k}$ be $(2^k-S_m)\times S_m$
matrices, and $\mathbf{G}^{(s)}_{m,k}=\begin{bmatrix}
\overline{\mathbf{G}}^{(s)}_{m,k} &
\underline{\mathbf{G}}^{(s)}_{m,k}
\end{bmatrix},s=0,1$.
\begin{gather*}
    \overline{\mathbf{G}}^{(1)}_{m,k}=
    \begin{bmatrix}
    1 & \mathbf{1}_{1\times k} & \cdots & \mathbf{1}_{1\times
    \binom{k}{n}} & \cdots & \mathbf{1}_{1\times \binom{k}{m}} \\
    \mathbf{0}_{k\times 1} & \mathbf{I}_{k} & \cdots &
    \mathbf{A}^1_{n,k} & \cdots & \mathbf{A}^1_{m,k} \\
    \vdots & \ddots & \ddots & \vdots & \ddots & \vdots\\
    \mathbf{0}_{\binom{k}{n}\times1} & \cdots
    & \mathbf{0} & \mathbf{I}_{\binom{k}{n}} & \cdots &
    \mathbf{A}^n_{m,k} \\
    \vdots & \vdots & \vdots & \ddots & \ddots & \vdots\\
    \mathbf{0}_{\binom{k}{m}\times1} & \mathbf{0}_{\binom{k}{m}\times
    k} & \cdots & \cdots  & \mathbf{0} &
    \mathbf{I}_{\binom{k}{m}} \\
    \end{bmatrix}\\
    \underline{\mathbf{G}}^{(0)}_{m,k}=(\overline{\mathbf{G}}^{(1)}_{m,k})^{RC}
\end{gather*}
are nonsingular and their inverse matrices become
\begin{gather}
    (\overline{\mathbf{G}}^{(1)}_{m,k})^{-1}=
    \begin{bmatrix}
    1 & -\mathbf{1}_{1\times k} & \cdots & (-1)^n\mathbf{1}_{1\times
    \binom{k}{n}} & \cdots & (-1)^m\mathbf{1}_{1\times \binom{k}{m}} \\
    \mathbf{0}_{k\times 1} & \mathbf{I}_{k} & \cdots &
    (-1)^{1+n}\mathbf{A}^1_{n,k} & \cdots & (-1)^{1+m}\mathbf{A}^1_{m,k} \\
    \vdots & \ddots & \ddots & \vdots & \ddots & \vdots\\
    \mathbf{0}_{\binom{k}{n}\times1} & \cdots
    & \mathbf{0} & \mathbf{I}_{\binom{k}{n}} & \cdots &
    (-1)^{n+m}\mathbf{A}^n_{m,k} \\
    \vdots & \vdots & \vdots & \ddots & \ddots & \vdots\\
    \mathbf{0}_{\binom{k}{m}\times1} & \mathbf{0}_{\binom{k}{m}\times
    k} & \cdots & \cdots  & \mathbf{0} &
    \mathbf{I}_{\binom{k}{m}}
    \end{bmatrix}\\
    (\underline{\mathbf{G}}^{(0)}_{m,k})^{-1}=(\overline{\mathbf{G}}^{(1)}_{m,k})^{-RR}
\end{gather}
\end{lemma}

Let
$\mathbf{p}^{(s)}_{m,k}[i]=\Pr(\mathbf{1}^{Co}_1=\mathbf{A}^1_{m,k}[1,i],\ldots,\mathbf{1}^{Co}_k=\mathbf{A}^1_{m,k}[k,i]|
    \mathbf{1}^{Rx}=s,\mathbf{1}^{Tx}=1)$, where $s=0,1$, $0\leq m\leq k$ and $1\leq i\leq\binom{k}{m}$, $\mathbf{1}^{Co}_i$
denotes spectrum availability at the $i$th cooperative node, and let
\begin{equation} \label{e19}
    \mathbf{P}^{(s)}_k=\begin{bmatrix}(\mathbf{p}^{(s)}_{0,k})^T &
    (\mathbf{p}^{(s)}_{1,k})^T & \cdots & (\mathbf{p}^{(s)}_{k,k})^T
    \end{bmatrix}^T
\end{equation}
Therefore, $\mathbf{P}^{(s)}_k$ characterizes joint pmf of spectrum
availability at $k$ cooperative nodes,
$\mathbf{1}^{Co}_1,\ldots,\mathbf{1}^{Co}_k$. Similarly, we divide
$\mathbf{P}^{(s)}_k$ into two parts. Let
$\overline{\mathbf{P}}^{(1)}_{m,k}$ and
$\underline{\mathbf{P}}^{(0)}_{m,k}$ be $S_m\times1$ vectors,
$\underline{\mathbf{P}}^{(1)}_{m,k}$ and
$\overline{\mathbf{P}}^{(0)}_{m,k}$ be $(2^k-S_m)\times1$ vectors,
and $\mathbf{P}^{(s)}_k=\begin{bmatrix}
(\overline{\mathbf{P}}^{(s)}_{m,k})^T &
(\underline{\mathbf{P}}^{(s)}_{m,k})^T
\end{bmatrix}^T,s=0,1$. In addition, let
\begin{equation*} 
    \mathbf{q}^{(s)}_{m,k}[j]=\Pr(\mathbf{1}^{Co}_{k_1}=s,\ldots,\mathbf{1}^{Co}_{k_m}=s|
    \mathbf{1}^{Rx}=s,\mathbf{1}^{Tx}=1,\{k_1,\ldots,k_m\}=\mathcal{I}_{m,k}(j))
\end{equation*}
$s=0,1$, $1\leq m\leq k$ and $1\leq j\leq \binom{k}{m}$, which
specifies the $m$th order marginal pmf of
$\mathbf{1}^{Co}_1,\ldots,\mathbf{1}^{Co}_k$. Arrange them into a
vector form, we define
\begin{equation}
    \mathbf{Q}^{(s)}_{m,k}=\begin{bmatrix}1 & (\mathbf{q}^{(s)}_{1,k})^T & (\mathbf{q}^{(s)}_{2,k})^T & \cdots &
    (\mathbf{q}^{(s)}_{m,k})^T\end{bmatrix}^T,s=0,1 \label{e7}
\end{equation}
Please note that
$\mathbf{q}^{(1)}_{k,k}=\mathbf{Q}^{(1)}_{k,k}[2^k]=\mathbf{P}^{(1)}_k[2^k]$
and
$\mathbf{q}^{(0)}_{k,k}=\mathbf{Q}^{(0)}_{k,k}[2^k]=\mathbf{P}^{(0)}_k[1]$.
\begin{lemma} \label{t2}
Marginal and joint pmf of spectrum availability at k cooperative
nodes satisfy
$\mathbf{G}^{(s)}_{m,k}\mathbf{P}^{(s)}_k=\mathbf{Q}^{(s)}_{m,k}$ or
\begin{align}
    \overline{\mathbf{P}}^{(1)}_{k,K}&=(\overline{\mathbf{G}}^{(1)}_{k,K})^{-1}
    (\mathbf{Q}^{(1)}_{k,K}-\underline{\mathbf{G}}^{(1)}_{k,K}\underline{\mathbf{P}}^{(1)}_{k,K})\\
    \underline{\mathbf{P}}^{(0)}_{k,K}&=(\underline{\mathbf{G}}^{(0)}_{k,K})^{-1}
    (\mathbf{Q}^{(0)}_{k,K}-\overline{\mathbf{G}}^{(0)}_{k,K}\overline{\mathbf{P}}^{(0)}_{k,K})
\end{align}
\end{lemma}

\begin{lemma} \label{cor3}
$\mathbf{P}^{(s)}_k$ provides equivalent information to
$\mathbf{Q}^{(s)}_{k,k},s=0,1$, or specifically,
\begin{equation} \label{e8}
    \mathbf{P}^{(s)}_k=\mathbf{c}^{(s)}_k+\mathbf{Q}^{(s)}_{k,k}[2^k]\mathbf{b}^{(s)}_k,s=0,1
\end{equation}
where
\begin{align*} 
    \mathbf{c}^{(1)}_k&=
    \begin{bmatrix}
    (\overline{\mathbf{G}}^{(1)}_{k-1,k})^{-1}\mathbf{Q}^{(1)}_{k-1,k}
    \\ 0
    \end{bmatrix} \qquad
    \mathbf{c}^{(0)}_k=
    \begin{bmatrix}
    0\\
    (\underline{\mathbf{G}}^{(0)}_{k-1,k})^{-1}\mathbf{Q}^{(0)}_{k-1,k}
    \end{bmatrix} \\
    \mathbf{b}^{(1)}_k&=
    \begin{bmatrix}
    (-1)^k\mathbf{1}_{\binom{k}{0}\times1}^T & \cdots &
    (-1)^{(k-j)}\mathbf{1}_{\binom{k}{j}\times1}^T & \cdots &
    (-1)^0\mathbf{1}_{\binom{k}{k}\times1}^T
    \end{bmatrix}^T\\
    \mathbf{b}^{(0)}_k&=(-1)^k\mathbf{b}^{(1)}_k 
\end{align*}
\end{lemma}

\subsection{Multiple Cooperative Nodes}
Assume there are $K$ cooperative nodes with corresponding spectrum
availability
$\mathbf{1}^{Co}_1,\mathbf{1}^{Co}_2,\ldots,\mathbf{1}^{Co}_{K}$ and
their joint pmf conditionally on $\mathbf{1}^{Rx}=s$ and
$\mathbf{1}^{Tx}=1$, $\mathbf{P}^{(s)}_K,s=0,1$.
\begin{prop} \label{p2}
Spectrum sensing with multiple cooperative nodes becomes
\begin{equation} \label{e14}
    \hat{\mathbf{1}}^{link}=\mathbf{1}^{Tx}\bigoplus_{j=1}^{2^K}{\mathbf{\Gamma}[j]\prod_{i=1}^{K}
    {[\mathbf{A}^1_K[i,j]\mathbf{1}^{Co}_k\oplus(1-\mathbf{A}^1_K[i,j])\bar{\mathbf{1}}^{Co}_k]}}
\end{equation}
and
\begin{equation}
    R\left(\mathbf{P}^{(1)}_K,\mathbf{P}^{(0)}_K\right)=
    \sum_{i=1}^{2^K}{\min\{w(1-\alpha)\mathbf{P}^{(0)}_K[i],\alpha \mathbf{P}^{(1)}_K[i]\}}
\end{equation}
where
$\mathbf{\Gamma}=\mathbf{1}_{[\alpha\mathbf{P}^{(1)}_K-w(1-\alpha)\mathbf{P}^{(0)}_K\geq0]}$
and $\oplus$ denotes OR operation.
\end{prop}
\begin{proof} By the likelihood ratio test, we have
\begin{equation*}
    \frac{\mathbf{P}^{(1)}_K[i]}{\mathbf{P}^{(0)}_K[i]}
    {\substack{\overset{\hat{\mathbf{1}}^{link}=1}{\geq}\\
    \underset{\hat{\mathbf{1}}^{link}=0}{<}}}
    \frac{w(1-\alpha)}{\alpha}
\end{equation*}
Therefore, the optimum detector becomes
\begin{equation*}
    \mathbf{\Gamma}[i]=\hat{\mathbf{1}}^{link}|_{\mathbf{1}^{Co}_1=\mathbf{A}^1_K[1,i],\cdots,\mathbf{1}^{Co}_K=\mathbf{A}^1_K[K,i]}
    =\mathbf{1}_{[\alpha\mathbf{P}^{(1)}_K[i]-w(1-\alpha)\mathbf{P}^{(0)}_K[i]\geq0]}
\end{equation*}
With the result and combining binary arithmetic we obtain
(\ref{e14}). In terms of Bayesian risk,
\begin{align*}
    R_{opt}&=\sum_{i=1}^{2^K}{[w(1-\alpha)\mathbf{P}^{(0)}_K[i]
    \Pr(\hat{\mathbf{1}}^{link}=1|\mathbf{1}^{Co}_1=\mathbf{A}^1_K[1,i],\cdots,\mathbf{1}^{Co}_K=\mathbf{A}^1_K[K,i],\mathbf{1}^{Rx}=0,\mathbf{1}^{Tx}=1)+}\\
    &\quad \alpha\mathbf{P}^{(1)}_K[i]
    \Pr(\hat{\mathbf{1}}^{link}=0|\mathbf{1}^{Co}_1=\mathbf{A}^1_K[1,i],\cdots,\mathbf{1}^{Co}_K=\mathbf{A}^1_K[K,i],\mathbf{1}^{Rx}=1,\mathbf{1}^{Tx}=1)]\\
    &=\sum_{i=1}^{2^K}{[w(1-\alpha)\mathbf{P}^{(0)}_K[i]\mathbf{\Gamma}[i]+
    \alpha\mathbf{P}^{(1)}_K[i](1-\mathbf{\Gamma}[i])]}\\
    &=\sum_{i=1}^{2^K}{\min\{w(1-\alpha)\mathbf{P}^{(0)}_K[i],\alpha \mathbf{P}^{(1)}_K[i]\}}
\end{align*}
\end{proof}

We consider two cooperative nodes insightfully to understand how
multiple cooperative nodes improve performance of spectrum sensing.
For spectrum availability at two cooperative nodes
$\mathbf{1}^{Co}_1,\mathbf{1}^{Co}_2$, let
$\mathbf{q}^{(1)}_{1,2}=\begin{bmatrix} \beta_1 & \beta_2
\end{bmatrix}^T$ and $\mathbf{q}^{(0)}_{1,2}=\begin{bmatrix} \gamma_1 & \gamma_2
\end{bmatrix}^T$. In addition, their joint probability is specified as follows. When $\mathbf{1}^{Rx}=1$,
PS is likely inactive and then $\mathbf{1}^{Co}_1$ and
$\mathbf{1}^{Co}_2$ are independent. On the other hand, when
$\mathbf{1}^{Rx}=0$, $\mathbf{1}^{Co}_1$ and $\mathbf{1}^{Co}_2$ are
correlated with correlation $\rho_{12}$. With the constraints on
$\mathbf{q}^{(0)}_{1,2}$, we have
\begin{equation} \label{e33}
    \mathbf{P}^{(0)}_2=
    \begin{bmatrix}
    \gamma_1\gamma_2+\Delta &
    (1-\gamma_1)\gamma_2-\Delta &
    \gamma_1(1-\gamma_2)-\Delta &
    (1-\gamma_1)(1-\gamma_2)+\Delta
    \end{bmatrix}^T
\end{equation}
where
$\Delta=\sqrt{\gamma_1\gamma_2(1-\gamma_1)(1-\gamma_2)}\rho_{12}$.

\subsubsection{Independent ($\rho_{12}=0$)}
In case $\mathbf{1}^{Co}_1$ and $\mathbf{1}^{Co}_2$ are
conditionally independent, it leads to conventional assumption in
cooperative spectrum sensing.
\begin{prop}
For two cooperative nodes with independent spectrum availability,
spectrum sensing becomes
\begin{equation} \label{e32}
    \hat{\mathbf{1}}^{link}=
    \begin{cases}
    \mathbf{1}^{Tx}, &\text{if $\alpha\geq\alpha_{(4)}$}\\
    \mathbf{1}^{Tx}\bigoplus_{i=1}^2{(\mathbf{1}_{[\rho_i>0]}\mathbf{1}^{Co}_i+\mathbf{1}_{[\rho_i<0]}\bar{\mathbf{1}}^{Co}_i)},
    &\text{if $\alpha_{(3)}\leq\alpha<\alpha_{(4)}$}\\
    \mathbf{1}^{Tx}(\mathbf{1}_{[\rho_k>0]}\mathbf{1}^{Co}_k+\mathbf{1}_{[\rho_k<0]}\bar{\mathbf{1}}^{Co}_k),
    &\text{if $\alpha_{(2)}<\alpha<\alpha_{(3)},k=\arg\max_iM_R(i)$} \\
    \mathbf{1}^{Tx}\prod_{i=1}^2{(\mathbf{1}_{[\rho_i>0]}\mathbf{1}^{Co}_i+\mathbf{1}_{[\rho_i<0]}\bar{\mathbf{1}}^{Co}_i)},
    &\text{if $\alpha_{(1)}<\alpha\leq\alpha_{(2)}$} \\
    0, &\text{if $\alpha\leq\alpha_{(1)}$}
    \end{cases}
\end{equation}
where
\begin{equation} \label{e36}
    (\alpha_{(1)},\alpha_{(2)},\alpha_{(3)},\alpha_{(4)})=
    \begin{cases}
    (\alpha^{C}_4,\alpha^{C}_2,\alpha^{C}_3,\alpha^{C}_1)
    &\text{if $\rho_1>0,\rho_2>0,\delta^+_{1}>\delta^+_{2}$}\\
    (\alpha^{C}_4,\alpha^{C}_3,\alpha^{C}_2,\alpha^{C}_1)
    &\text{if $\rho_1>0,\rho_2>0,\delta^+_{1}<\delta^+_{2}$}\\
    (\alpha^{C}_2,\alpha^{C}_4,\alpha^{C}_1,\alpha^{C}_3)
    &\text{if $\rho_1>0,\rho_2<0,\delta^+_{1}>\delta^-_{2}$}\\
    (\alpha^{C}_2,\alpha^{C}_1,\alpha^{C}_4,\alpha^{C}_3)
    &\text{if $\rho_1>0,\rho_2<0,\delta^+_{1}<\delta^-_{2}$}\\
    (\alpha^{C}_3,\alpha^{C}_1,\alpha^{C}_4,\alpha^{C}_2)
    &\text{if $\rho_1<0,\rho_2>0,\delta^-_{1}>\delta^+_{2}$}\\
    (\alpha^{C}_3,\alpha^{C}_4,\alpha^{C}_1,\alpha^{C}_2)
    &\text{if $\rho_1<0,\rho_2>0,\delta^-_{1}<\delta^+_{2}$}\\
    (\alpha^{C}_1,\alpha^{C}_3,\alpha^{C}_2,\alpha^{C}_4)
    &\text{if $\rho_1<0,\rho_2<0,\delta^-_{1}>\delta^-_{2}$}\\
    (\alpha^{C}_1,\alpha^{C}_2,\alpha^{C}_3,\alpha^{C}_4)
    &\text{if $\rho_1<0,\rho_2<0,\delta^-_{1}<\delta^-_{2}$}\\
    \end{cases}
\end{equation}
$\alpha^{C}_i=w\mathbf{P}^{(0)}_2[i]/(\mathbf{P}^{(1)}_2[i]+w\mathbf{P}^{(0)}_2[i])$,
$M_R(i)=\mathbf{1}_{[\rho_i\geq0]}\delta^+_{i}+\mathbf{1}_{[\rho_i<0]}\delta^-_{i}$,
$\delta^+_{i}=\gamma_i\beta_i/((1-\gamma_i)(1-\beta_i))$,
$\delta^-_{i}=(1-\gamma_i)(1-\beta_i)/(\gamma_i\beta_i)$, and
$\rho_i$ denotes the correlation between $\mathbf{1}^{Co}_i$ and
$\mathbf{1}^{Rx}$.
\end{prop}

We list all valid orders of
$(\alpha^{C}_1,\alpha^{C}_2,\alpha^{C}_3,\alpha^{C}_4)$ in
(\ref{e36}) and the the spectrum sensing is determined according to
the value of $\alpha$ by arguing (\ref{e14}). We observe that
$\alpha$ is high ($\alpha_{(3)}\leq\alpha<\alpha_{(4)}$), any one of
cooperative nodes helps the spectrum sensing, which leads the
spectrum sensing to OR operation. However, when $\alpha$ is low
($\alpha_{(1)}<\alpha\leq\alpha_{(2)}$), CR-Tx requires more
evidence to claim available link and the spectrum sensing becomes
AND operation. In addition, it is interesting to note that there
exists a region of $\alpha$ ($\alpha_{(2)}<\alpha<\alpha_{(3)}$)
such that CR-Tx only depends on one of two cooperative nodes, which
motivates us to define a metric or a measure to evaluate cooperative
nodes.

\begin{defi}
\textbf{Reliability} of a cooperative node is measured by $M_R$.
$\mathbf{1}^{Co}_i$ is said to be more or equally \textbf{reliable}
than (to) $\mathbf{1}^{Co}_j$ if $M_R(i)\geq M_R(j)$, which is
denoted by $\mathbf{1}^{Co}_i\unrhd\mathbf{1}^{Co}_j$.
\end{defi}

\begin{prop}
For $K$ cooperative nodes with independent spectrum availability,
without loss of generality, we assume $\rho_i\geq0$ for
$i=1,\ldots,n$ and $\rho_i<0$ for $i=n+1,\ldots,K$. Spectrum sensing
becomes $\hat{\mathbf{1}}^{link}=\mathbf{1}^{Tx}\mathbf{1}_{[s]}$,
where
\begin{equation} \label{e34}
    s=\left\{\sum_{i\in\mathcal{C}^+\cup\mathcal{C}^-}{\log M_R(i)}
    \geq\log\left(\frac{w(1-\alpha)}{\alpha}\right)
    +\sum_{i=1}^{n}{\log\left(\frac{\gamma_i}{1-\beta_i}\right)}
    +\sum_{i=n+1}^{K}{\log\left(\frac{1-\gamma_i}{\beta_i}\right)}\right\}
\end{equation}
and $\mathcal{C}^+=\{i|\mathbf{1}^{Co}_i=1,i=1,\ldots,n\}$,
$\mathcal{C}^-=\{i|\mathbf{1}^{Co}_i=0,i=n+1,\ldots,K\}$.
\end{prop}
\begin{proof}
Since
$\mathbf{1}^{Co}_1,\mathbf{1}^{Co}_2,\ldots,\mathbf{1}^{Co}_{K}$ are
independent, the likelihood ratio test becomes
\begin{align*}
    \frac{\mathbf{P}^{(1)}_K[i]}{\mathbf{P}^{(0)}_K[i]}
    &=\prod_{i\in\mathcal{C}^+}\frac{\beta_i}{1-\gamma_i}
    \prod_{i\in\{1,\ldots,n\}\setminus\mathcal{C}^+}\frac{1-\beta_i}{\gamma_i}
    \prod_{i\in\mathcal{C}^-}\frac{1-\beta_i}{\gamma_i}
    \prod_{i\in\{n+1,\ldots,K\}\setminus\mathcal{C}^-}\frac{\beta_i}{1-\gamma_i}\\
    &=\prod_{i\in\mathcal{C}^+}\frac{\beta_i\gamma_i}{(1-\beta_i)(1-\gamma_i)}
    \prod_{i\in\{1,\ldots,n\}}\frac{1-\beta_i}{\gamma_i}
    \prod_{i\in\mathcal{C}^-}\frac{(1-\beta_i)(1-\gamma_i)}{\beta_i\gamma_i}
    \prod_{i\in\{n+1,\ldots,K\}}\frac{\beta_i}{1-\gamma_i}\\
    &=\prod_{i\in\mathcal{C}^+\cup\mathcal{C}^-}M_R(i)
    \prod_{i\in\{1,\ldots,n\}}\frac{1-\beta_i}{\gamma_i}
    \prod_{i\in\{n+1,\ldots,K\}}\frac{\beta_i}{1-\gamma_i}
    {\substack{\overset{\hat{\mathbf{1}}^{link}=1}{\geq}\\
    \underset{\hat{\mathbf{1}}^{link}=0}{<}}}
    \frac{w(1-\alpha)}{\alpha}
\end{align*}
Taking logarithm at both sides and rearranging the formula, we
complete the proof.
\end{proof}

We observe that reliability $M_R(i)$ is used to quantify the
information of $\mathbf{1}^{Rx}$ provided by the $i$th cooperative
node. Please note that if $\mathbf{1}^{Co}_i$ is independent of
$\mathbf{1}^{Rx}$, $M_R(i)=1$ and $\mathbf{1}^{Co}_i$ is irrelevant
to the spectrum sensing. Therefore, reliability can imply sensing
capability, that is, one cooperative node with higher reliability
has better sensing capability. Reliability can thus serve a
criterion to select cooperative nodes when number of cooperative
nodes is limited due to appropriate overhead caused by information
exchange. Specially, if there are $K$ equally reliable cooperative
nodes, each cooperative node provides equal amount of information
about $\mathbf{1}^{Rx}$ and the spectrum sensing rule turns out to
be Counting rule. This is a generalization from identically and
independently distributed (i.i.d.) observations
\cite{GSSinCRN:Viswanathan89} in conventional distributed detection
to equally reliable observations.

\subsubsection{Correlated ($\rho_{12}\neq0$)}
When there exists correlation between spectrum availability at
cooperative nodes, the joint probabilities $\mathbf{P}^{(0)}_2$ are
shifted by $\Delta$, as in (\ref{e33}). For example, if the
correlation is positive $\alpha^{C}_1$ and $\alpha^{C}_4$ increase
while $\alpha^{C}_2$ and $\alpha^{C}_3$ decrease. If the correlation
increases, eventually, the order of
$(\alpha^{C}_1,\alpha^{C}_2,\alpha^{C}_3,\alpha^{C}_4)$ will switch
and the spectrum sensing in (\ref{e32}) will change accordingly.
However, the correlated case becomes tedious and we consider a
simple but meaningful example, where cooperative nodes have
symmetric error rates (i.e. $\beta_i=\gamma_i$) and reliability
becomes
$\mathbf{1}_{[\rho_i\geq0]}\beta_i+\mathbf{1}_{[\rho_i<0]}(1-\beta_i)$.
Similarly, with the aid of (\ref{e14}) and (\ref{e33}), the spectrum
sensing can be easily derived.

\begin{prop} \label{p3}
For two cooperative nodes with correlated spectrum availability and
symmetric error rate satisfying $\rho_1>0, \rho_2>0$, and
$\mathbf{1}^{Co}_1\rhd\mathbf{1}^{Co}_2$, the spectrum sensing would
be (\ref{e32}) with modifications according to $\Delta$.
\begin{equation} \label{e40}
    \hat{\mathbf{1}}^{link}=
    \begin{cases}
    \mathbf{1}^{Tx}(\mathbf{1}^{Co}_1\mathbf{1}^{Co}_2\oplus\bar{\mathbf{1}}^{Co}_1\bar{\mathbf{1}}^{Co}_2)
    &\text{if
    $\alpha_{(2)}<\alpha<\alpha_{(3)},\Delta<(1-2\beta_1)\beta_2$}\\
    \mathbf{1}^{Tx}(\mathbf{1}^{Co}_1\oplus\bar{\mathbf{1}}^{Co}_2)
    &\text{if
    $\alpha_{(3)}<\alpha<\alpha_{(4)},\Delta<(1-2\beta_2)\beta_1$}\\
    \mathbf{1}^{Tx}\mathbf{1}^{Co}_1\bar{\mathbf{1}}^{Co}_2
    &\text{if
    $\alpha_{(1)}<\alpha<\alpha_{(2)},\Delta\geq(2\beta_2-1)(1-\beta_1)$}\\
    \mathbf{1}^{Tx}(\mathbf{1}^{Co}_1\otimes\mathbf{1}^{Co}_2)
    &\text{if $\alpha_{(2)}<\alpha<\alpha_{(3)},\Delta\geq(2\beta_1-1)(1-\beta_2)$}
    \end{cases}
\end{equation}
where
\begin{equation} \label{e37}
    (\alpha_{(1)},\alpha_{(2)},\alpha_{(3)},\alpha_{(4)})=
    \begin{cases}
    (\alpha^{C}_4,\alpha^{C}_1,\alpha^{C}_2,\alpha^{C}_3)
    &\text{if $\Delta_{min}\leq\Delta<(1-2\beta_1)\beta_2,\beta_1(1+\beta_2)<1$}\\
    (\alpha^{C}_4,\alpha^{C}_2,\alpha^{C}_1,\alpha^{C}_3)
    &\text{if $(1-2\beta_1)\beta_2\leq\Delta<(1-2\beta_2)\beta_1$}\\
    (\alpha^{C}_4,\alpha^{C}_2,\alpha^{C}_3,\alpha^{C}_1)
    &\text{if $(1-2\beta_2)\beta_1\leq\Delta<(2\beta_2-1)(1-\beta_1)$}\\
    (\alpha^{C}_2,\alpha^{C}_4,\alpha^{C}_3,\alpha^{C}_1)
    &\text{if $(2\beta_2-1)(1-\beta_1)\leq\Delta<(2\beta_1-1)(1-\beta_2)$}\\
    (\alpha^{C}_2,\alpha^{C}_3,\alpha^{C}_4,\alpha^{C}_1)
    &\text{if $(2\beta_1-1)(1-\beta_2)\leq\Delta\leq\Delta_{max},\beta_1(2-\beta_2)<1$}\\
    \end{cases}
\end{equation}
$\Delta_{min}=-(1-\beta_1)(1-\beta_2)$,
$\Delta_{max}=(1-\beta_1)\beta_2$, and $\otimes$ denotes XOR
operation.
\end{prop}

All possible switching orders of
$(\alpha^{C}_1,\alpha^{C}_2,\alpha^{C}_3,\alpha^{C}_4)$ according to
$\Delta$ are listed in (\ref{e37}) and the first and the last orders
are impossible unless an additional condition is satisfied to make
the regions of $\Delta$ valid, i.e. $\beta_1(1+\beta_2)<1$ and
$\beta_1(2-\beta_2)<1$ respectively. Since $\mathbf{1}^{Co}_1$ and
$\mathbf{1}^{Co}_2$ are correlated when $\mathbf{1}^{Rx}=0$ (i.e.
$\rho_{12}\neq0$ or $\Delta\neq0$), not only $\mathbf{1}^{Co}_1$ and
$\mathbf{1}^{Co}_2$ alone but also the identity of
$\mathbf{1}^{Co}_1$ and $\mathbf{1}^{Co}_2$ (i.e.
$\mathbf{1}^{Co}_1\otimes\mathbf{1}^{Co}_2$ or
$\mathbf{1}^{Co}_1\mathbf{1}^{Co}_2\oplus\bar{\mathbf{1}}^{Co}_1\bar{\mathbf{1}}^{Co}_2$)
can provide information about $\mathbf{1}^{Rx}$ and thus can be used
to determine CR link availability. This is actually similar to
covariance-based detection. For example, if
$\Delta\geq(2\beta_2-1)(1-\beta_1)\geq0$, $\mathbf{1}^{Co}_1$ and
$\mathbf{1}^{Co}_2$ are probably identical when spectrum is
unavailable at CR-Rx, i.e. $\mathbf{1}^{RX}=0$, and
$\mathbf{1}^{Co}_2$ in (\ref{e32}) is then replaced by
$\mathbf{1}^{Co}_1\otimes\mathbf{1}^{Co}_2$. Furthermore, when
$\Delta$ increases and is greater than $(2\beta_1-1)(1-\beta_2)$,
the roles of $\mathbf{1}^{Co}_1$ and
$\mathbf{1}^{Co}_1\otimes\mathbf{1}^{Co}_2$ switch because the
identity of $\mathbf{1}^{Co}_1$ and $\mathbf{1}^{Co}_2$ can provide
more information about $\mathbf{1}^{RX}$ than $\mathbf{1}^{Co}_1$
alone. Alternatively, when $\Delta<0$, the results can be similarly
explained. In addition, it is interesting to note that even if
$\mathbf{1}^{Co}_2$ is independent to $\mathbf{1}^{Rx}$ (i.e.
$\beta_2=1/2$), $\mathbf{1}^{Co}_2$ may become helpful due to the
correlation between $\mathbf{1}^{Co}_1$ and $\mathbf{1}^{Co}_2$. In
the next section, we will further investigate impacts of correlation
between $\mathbf{1}^{Co}_1$ and $\mathbf{1}^{Co}_2$ on network
operation.

\subsection{Multiple Cooperative Nodes with Limited Statistical Information}
In CRN or self-organizing networks, due to lacking of centralized
coordination, each node in CRN can only sense and exchange local
information. In addition, dynamic wireless channels and mobility of
nodes make the situation severer and one node can only acquire
information within limited sensing duration. We can either design
systems under simplified assumptions, which may result in severe
performance degradation, or apply advanced signal processing
techniques based on minimax criterion \cite{GSSinCRN:Poor94}, robust
to outliers in networks, as we are going to do hereafter.

To derive the optimum Bayesian detection in \textbf{Proposition
\ref{p2}}, we have to acquire joint pmf of spectrum availability at
$K$ cooperative nodes, which may require long observation interval
to achieve acceptable estimation error. If we only have up to the
$k$th order marginal pmf (related to capability of observation), i.e
$\mathbf{Q}^{(s)}_{k,K},s=0,1$ according to \textbf{Lemma
\ref{cor3}}, our design criterion becomes minimax criterion, that
is, we find the least-favorable joint pmf $\mathbf{P}^{(s)}_K,s=0,1$
such that maximizes Bayesian risk and then conduct the optimum
Bayesian detection under that joint probability. Therefore, the
problem can be formulated as follows.
\begin{prop} [Robust Cooperative Sensing]
Cooperative spectrum sensing with limited statistical information
$\mathbf{Q}^{(s)}_{k,K},s=0,1$ becomes (\ref{e14}) with
$(\mathbf{P}^{(1)}_K,\mathbf{P}^{(0)}_K)$ replaced by
$(\mathbf{P}^{(1)}_{opt},\mathbf{P}^{(0)}_{opt})$, where
\begin{equation}
\begin{aligned} \label{e35}
    (\mathbf{P}^{(1)}_{opt},\mathbf{P}^{(0)}_{opt})&=
    \arg\max_{\mathbf{P}^{(1)}_K,\mathbf{P}^{(0)}_K}
    R\left(\mathbf{P}^{(1)}_K,\mathbf{P}^{(0)}_K\right)=
    \arg\min_{\mathbf{P}^{(1)}_K,\mathbf{P}^{(0)}_K}
    \|w(1-\alpha)\mathbf{P}^{(0)}_K-\alpha\mathbf{P}^{(1)}_K\|_1\\
    \text{s.t.}\quad
    \overline{\mathbf{P}}^{(1)}_{k,K}&=(\overline{\mathbf{G}}^{(1)}_{k,K})^{-1}
    (\mathbf{Q}^{(1)}_{k,K}-\underline{\mathbf{G}}^{(1)}_{k,K}\underline{\mathbf{P}}^{(1)}_{k,K})\\
    \underline{\mathbf{P}}^{(0)}_{k,K}&=(\underline{\mathbf{G}}^{(0)}_{k,K})^{-1}
    (\mathbf{Q}^{(0)}_{k,K}-\overline{\mathbf{G}}^{(0)}_{k,K}\overline{\mathbf{P}}^{(0)}_{k,K})\\
    \mathbf{0}_{2^K\times1}&\preceq\mathbf{P}^{(s)}_K\preceq\mathbf{1}_{2^K\times1},s=0,1
\end{aligned}
\end{equation}
\end{prop}
The last equality in the objective function is based on the fact
that $\min(x,y)=(x+y-|x-y|)/2$ and that the sum of probability
distribution is equal to one. The result is reasonable because in
order to minimize the objective function, the likelihood ratio
$\mathbf{P}^{(1)}_K[i]/\mathbf{P}^{(0)}_K[i]$ approaches to the
optimum threshold $w(1-\alpha)/\alpha$, which induces poor
performance of the detector and therefore increases Bayesian risk.
Furthermore, we could apply \textbf{Lemma \ref{t2}} to set the
constraints on joint pmf. Since vector norm is a convex function,
the problem can be solved by well-developed algorithms in convex
optimization \cite{GSSinCRN:Boyd04}.

In last part, we proposed a simple methodology to select cooperative
nodes based on reliability under assumption of independent
observations. However, in practice, there exists correlation among
spectrum availability at cooperative nodes and spectrum sensing may
change as we showed in \textbf{Proposition \ref{p3}}. In addition,
since the statistical information is limited within reasonable
observation interval, CR-Tx can select cooperative nodes to minimize
maximum Bayesian risk by minimax criterion.

\section{Application to Realistic Operation of CRN}
In preceding sections, we only considered single CR link in CRN.
However, CRN is not just a link level technology if we want to
successfully route packets from source to destination through CRs
and PS. In the following, we suggest a simple physical layer model
for CRN and investigate the impacts of spectrum sensing on network
operation and the role of a cooperative node playing in CRN, which
is impossible to be revealed from traditional treatment of spectrum
sensing. Since spectrum sensing may not be ideal and there exists
hidden terminal problem, we further define the true state for PS.
\begin{defi}
The true state for PS can be represented by the indicator
\begin{equation*}
    \mathbf{1}^{PS}=
    \begin{cases}
    1, &\text{PS either does not exist or is inactive}\\
    0, &\text{PS exists and is active}
    \end{cases}
\end{equation*}
\end{defi}
Therefore, with the definition of $\alpha$, we have
\begin{equation} \label{e38}
    \alpha=\frac{\sum_{s=0}^{1}{\Pr(\mathbf{1}^{Tx}=1,\mathbf{1}^{Rx}=1|\mathbf{1}^{PS}=s)\Pr(\mathbf{1}^{PS}=s)}}
    {\sum_{s=0}^{1}{\Pr(\mathbf{1}^{Tx}=1|\mathbf{1}^{PS}=s)\Pr(\mathbf{1}^{PS}=s)}}
\end{equation}
To connect relations between indicator functions of link
availability and realistic operation of CRN, we propose a simple
received power model.

\subsection{Received Power Model}
We model the received power from PS and background noise as
log-normal distribution, or $10log_{10}(P_S)\sim
N(\mu_S,\sigma^2_S)$ and $10log_{10}(P_N)\sim N(\mu_0,\sigma^2_0)$,
where $\sigma^2_S$ and $\sigma^2_0$ are used to quantify the
measurement uncertainty of the received power from PS and noise
respectively. In addition, $\mu_0$ is a constant whereas $\mu_S$
should be varied according to path loss and shadowing. More
specifically, let $\mu_S=K_0-10alog_{10}(d_{CR})-b_{CR}$, where
$K_0$ is a constant, $d_{CR}$ denotes distance from CR (either CR-Tx
or CR-Rx) to PS as in Fig \ref{Fig_2}, $a$ means path loss exponent,
and $b_{CR}$ represents shadowing effect. When $\mathbf{1}^{PS}=1$,
the received signal only comes from noise. However, when
$\mathbf{1}^{PS}=0$, the received signal is the superposition of
signal from PS and noise, which results in addition of two
log-normal random variables. We could simply model the received
power as another log-normal random variable with parameters
$\mu_{CR}$ and $\sigma_{CR}$, and under assumption of
$\sigma_S>\sigma_0$, we have
\begin{align} \label{e42}
    \mu_{CR}&=
    \begin{cases}
    \mu_0, &\text{if $\mu_S\leq\mu_0-\sigma_S$}\\
    \mu_S, &\text{if $\mu_S\geq\mu_0+\sigma_S$}\\
    (\mu_S+\mu_0+\sigma_S)/2, &\text{otherwise}
    \end{cases}\\
    \sigma_{CR}^2&=
    \begin{cases}
    \sigma^2_0, &\text{if $\mu_S\leq\mu_0-\sigma_S$}\\
    \sigma^2_S, &\text{if $\mu_S\geq\mu_0+2\sigma_S$}\\
    \frac{\sigma^2_S-\sigma^2_0}{3\sigma_S}(\mu_S-\mu_0)+
    \frac{\sigma^2_S+2\sigma^2_0}{3}, &\text{otherwise}
    \end{cases} \label{e43}
\end{align}
By simulation, the distribution of the simplified model, although
not exactly identical to, is close to the simulated distribution,
especially in terms of mean and variance. It justifies our
simplified model.

\subsection{Spectrum Sensing at CR-Tx and Reception at CR-Rx}
Recall the conditions that CRs can successfully communicate over a
link. Assume CR-Tx adopts an energy detector in the hypothesis
testing (\ref{e27}) and there is no interference from co-existing
systems. The detector can be represented as
\begin{equation*}
    P_{Tx}{\substack{\overset{\mathbf{1}^{Tx}=1}{\leq}\\
    \underset{\mathbf{1}^{Tx}=0}{>}}}\tau_{Tx} \quad \text{(in dB)}
\end{equation*}
where $P_{Tx}$ denotes the received power at CR-Tx and $\tau_{Tx}$
is a fixed threshold since the detector is designed under a given
SINR. On the other hand, to successfully receive packets, the SINR
at CR-Rx should be greater than minimum value $\eta_{outage}$ as
shown in (\ref{e28}). Similarly, spectrum availability at CR-Rx can
be represented as
\begin{equation*}
    P_{Rx}{\substack{\overset{\mathbf{1}^{Rx}=1}{\leq}\\
    \underset{\mathbf{1}^{Rx}=0}{>}}}\tau_{Rx} \quad \text{(in dB)}
\end{equation*}
where $P_{Rx}$ denotes the received power from PS and noise at
CR-Rx. Different from CR-Tx, $\tau_{Rx}$ is varied according to the
received power from CR-Tx. For simplicity, we only consider
propagation loss in modeling the received power from CR-TX and have
$\tau_{Rx}=L_0-10alog_{10}(r_{Rx})$, where $L_0$ is a constant and
$r_{Rx}$ denotes distance between CR-Tx and CR-Rx.

We suppose that the measurement uncertainties and hence the received
power from PS and noise at CR-Tx and CR-Rx are independent. However,
to model spatial behavior for CR-Tx and CR-Rx, we consider the
relation of shadowing between CR-Tx and CR-Rx. Intuitively, the
relation should depend on locations of CR-Tx, CR-Rx, PS, along with
the obstacle size and we proceed based on a linear model
\begin{equation} \label{e29}
    b_{Rx}=
    \begin{cases}
    \max\left\{b_{Tx}(1-\frac{r_{Rx}}{2\kappa}),0\right\},
    &\text{if $r_{Rx}\cos(\theta_{Rx})\leq d_{Tx}$} \\
    0, &\text{if $r_{Rx}\cos(\theta_{Rx})> d_{Tx}$}
    \end{cases}
\end{equation}
where $\kappa$ denotes parameter of obstacle size and $\theta_{Rx}$
is the angle between line segments with starting point at CR-Tx and
end points at CR-Rx and PS, as shown in Fig. \ref{Fig_2} for
illustration. In this model, shadowing at CR-Rx $b_{Rx}$ linearly
decreases with respect to the distance between CR-Tx and CR-Rx
$r_{Rx}$ with rate inverse proportional to the obstacle size
$\kappa$ and is equal to zero when CR-Rx is far apart from CR-Tx or
PS is located in the middle of CR-Tx and CR-Rx. Additionally, since
shadowing parameter achieves maximum at CR-Tx, this results in the
worst case scenario in spectrum sensing. Finally, from log-normal
fading distribution,
\begin{equation*}
    \alpha=\frac{\Pr(\mathbf{1}^{PS}=1)Q\left(\frac{\mu_0-\tau_{Tx}}{\sigma_0}\right)
    Q\left(\frac{\mu_0-\tau_{Rx}}{\sigma_0}\right)+
    \Pr(\mathbf{1}^{PS}=0)Q\left(\frac{\mu_{Tx}-\tau_{Tx}}{\sigma_{Tx}}\right)
    Q\left(\frac{\mu_{Rx}-\tau_{Rx}}{\sigma_{Rx}}\right)}
    {\Pr(\mathbf{1}^{PS}=1)Q\left(\frac{\mu_0-\tau_{Tx}}{\sigma_0}\right)
    +\Pr(\mathbf{1}^{PS}=0)Q\left(\frac{\mu_{Tx}-\tau_{Tx}}{\sigma_{Tx}}\right)}
\end{equation*}
where $Q(x)$ denotes the right-tail probability of a Gaussian random
variable with zero mean and unit variance.

\subsection{Cooperative Spectrum Sensing}
Under above proposed signal model, the analysis can be easily
extended to cooperative spectrum sensing. Considering a cooperative
node conducting an energy detector, we have
\begin{equation*}
    P_{Co}{\substack{\overset{\mathbf{1}^{Co}=1}{\leq}\\
    \underset{\mathbf{1}^{Co}=0}{>}}}\tau_{Co} \quad \text{(in dB)}
\end{equation*}
Furthermore, the correlation due to geography is established similar
to (\ref{e29}). We could therefore calculate $\beta$ and $\gamma$
similar to $\alpha$, as
\begin{align*}
    \beta&=\frac{\sum_{s=0}^{1}{\Pr(\mathbf{1}^{Tx}=1,\mathbf{1}^{Rx}=1,\mathbf{1}^{Co}=1|\mathbf{1}^{PS}=s)\Pr(\mathbf{1}^{PS}=s)}}
    {\sum_{s=0}^{1}{\Pr(\mathbf{1}^{Tx}=1,\mathbf{1}^{Rx}=1|\mathbf{1}^{PS}=s)\Pr(\mathbf{1}^{PS}=s)}} \\
    \gamma&=\frac{\sum_{s=0}^{1}{\Pr(\mathbf{1}^{Tx}=1,\mathbf{1}^{Rx}=0,\mathbf{1}^{Co}=0|\mathbf{1}^{PS}=s)\Pr(\mathbf{1}^{PS}=s)}}
    {\sum_{s=0}^{1}{\Pr(\mathbf{1}^{Tx}=1,\mathbf{1}^{Rx}=0|\mathbf{1}^{PS}=s)\Pr(\mathbf{1}^{PS}=s)}}
\end{align*}
With the relation between statistical information
$\{\alpha,\beta,\gamma\}$ and received power model, we can
mathematically determine allowable transmission region of CR-Tx.

\subsection{Neighborhood of CR-Tx}
In Section III and IV, we have developed spectrum sensing under
different assumptions and note that spectrum sensing depends on the
value of $\alpha$, i.e., spatial behavior of CR-Tx and CR-Rx.
Especially, there is even a region of $\alpha$ that CR-Tx is
prohibited from forwarding packets to CR-Rx and the link from CR-Tx
to CR-Rx is disconnected (i.e. $\hat{\mathbf{1}}^{link}=0$). This
undesirable phenomenon alters CRN topology and heavily affects
network performance, such as throughput of CRN, etc. Therefore, we
would like to theoretically study link properties in CRN and first
define the regions of $\alpha$ as follows.
\begin{defi}
The set $\{\alpha|\Pr(\hat{\mathbf{1}}^{link}=0)=1\}$ is called
\textbf{prohibitive region} while
$\{\alpha|\Pr(\hat{\mathbf{1}}^{link}=1)\neq0\}$ is called
\textbf{admissive region}. The boundary between these two sets is
called \textbf{critical boundary} of $\alpha$ and is denoted by
$\alpha_C$. Therefore,
$\{\alpha|\hat{\mathbf{1}}^{link}=0\}=\{\alpha|0\leq\alpha<\alpha_C\}$
and
$\{\alpha|\hat{\mathbf{1}}^{link}=1\}=\{\alpha|\alpha_C\leq\alpha\leq1\}$.
\end{defi}
If $\alpha$ lies in the prohibitive region, the link from CR-Tx to
CR-Rx is disconnected. The property and the engineering meaning of
$\alpha_C$ are addressed as follows.
\begin{lemma} \label{l6}
$\alpha_C$ is a decreasing function with respect to number of
cooperative nodes.
\end{lemma}
\begin{proof}
It is easy to show that for fixed $w$, $\alpha$ decreases as the
threshold of the likelihood ratio test $w(1-\alpha)/\alpha$
increases. Therefore, $\alpha_C$ can be determine by the largest
likelihood ratio. Assume the largest likelihood ratio with $k-1$
cooperative nodes occurs at $i_{max}$, i.e.,
$i_{max}=\arg\max_i\{\mathbf{P}^{(1)}_{k-1}[i]/\mathbf{P}^{(0)}_{k-1}[i]\}$.
When the $k$th cooperative node enters, let
\begin{align*}
    \tilde{\beta}_k&=\Pr(\mathbf{1}^{Co}_k=1|\mathbf{1}^{Co}_1=\mathbf{A}^1_{m,k-1}[1,i_{max}],\ldots,
    \mathbf{1}^{Co}_{k-1}=\mathbf{A}^1_{m,k-1}[k-1,i_{max}],\mathbf{1}^{Rx}=1,\mathbf{1}^{Tx}=1)\\
    \tilde{\gamma}_k&=\Pr(\mathbf{1}^{Co}_k=0|\mathbf{1}^{Co}_1=\mathbf{A}^1_{m,k-1}[1,i_{max}],\ldots,
    \mathbf{1}^{Co}_{k-1}=\mathbf{A}^1_{m,k-1}[k-1,i_{max}],\mathbf{1}^{Rx}=0,\mathbf{1}^{Tx}=1)
\end{align*}
Then, there are two likelihood ratio with $k$ cooperative nodes, say
$i$th and $j$th, becoming
\begin{align*}
    \frac{\mathbf{P}^{(1)}_{k}[i]}{\mathbf{P}^{(0)}_{k}[i]}&=
    \frac{\mathbf{P}^{(1)}_{k-1}[i_{max}]\tilde{\beta}_k}
    {\mathbf{P}^{(0)}_{k-1}[i_{max}](1-\tilde{\gamma}_k)}\\
    \frac{\mathbf{P}^{(1)}_{k}[j]}{\mathbf{P}^{(0)}_{k}[j]}&=
    \frac{\mathbf{P}^{(1)}_{k-1}[i_{max}](1-\tilde{\beta}_k)}
    {\mathbf{P}^{(0)}_{k-1}[i_{max}]\tilde{\gamma}_k}
\end{align*}
Since either $\tilde{\beta}_k+\tilde{\gamma}_k\geq1$ or
$\tilde{\beta}_k+\tilde{\gamma}_k<1$, one of the $i$th and the $j$th
likelihood ratio is not less than
$\mathbf{P}^{(1)}_{k-1}[i]/\mathbf{P}^{(0)}_{k-1}[i]$, which results
in lower $\alpha_C$.
\end{proof}
\begin{lemma} \label{l7}
The following two statements are equivalent:
\begin{enumerate}
    \item $\alpha\geq\alpha_C$
    \item $\Pr(\mathbf{1}^{link}=1|\hat{\mathbf{1}}^{link}=1)\geq w/(w+1)$
\end{enumerate}
\end{lemma}
\begin{proof} Since $\alpha\geq\alpha_C$ if and only if
$\Pr(\hat{\mathbf{1}}^{link}=1)\neq0$, we have
\begin{align*}
    &\Pr(\mathbf{1}^{link}=1|\hat{\mathbf{1}}^{link}=1) \\
    &=\Pr(\mathbf{1}^{Tx}=1,\mathbf{1}^{Rx}=1|\mathbf{1}^{Tx}=1,\hat{\mathbf{1}}^{Rx}=1)
    \\
    &=\frac{\Pr(\mathbf{1}^{Rx}=1,\hat{\mathbf{1}}^{Rx}=1|\mathbf{1}^{Tx}=1)}
    {\Pr(\hat{\mathbf{1}}^{Rx}=1|\mathbf{1}^{Tx}=1)}\\
    &=\frac{\Pr(\mathbf{1}^{Rx}=1|\mathbf{1}^{Tx}=1)\Pr(\hat{\mathbf{1}}^{Rx}=1|\mathbf{1}^{Rx}=1,\mathbf{1}^{Tx}=1)}
    {\sum_{s=0}^{1}{\Pr(\mathbf{1}^{Rx}=s|\mathbf{1}^{Tx}=1)\Pr(\hat{\mathbf{1}}^{Rx}=1|\mathbf{1}^{Rx}=s,\mathbf{1}^{Tx}=1)}} \\
    &=\frac{\frac{\alpha}{1-\alpha}\frac{\Pr(\hat{\mathbf{1}}^{Rx}=1|\mathbf{1}^{Rx}=1,\mathbf{1}^{Tx}=1)}
    {\Pr(\hat{\mathbf{1}}^{Rx}=1|\mathbf{1}^{Rx}=0,\mathbf{1}^{Tx}=1)}}
    {1+\frac{\alpha}{1-\alpha}\frac{\Pr(\hat{\mathbf{1}}^{Rx}=1|\mathbf{1}^{Rx}=1,\mathbf{1}^{Tx}=1)}
    {\Pr(\hat{\mathbf{1}}^{Rx}=1|\mathbf{1}^{Rx}=0,\mathbf{1}^{Tx}=1)}}\\
    &\geq\frac{\frac{\alpha}{1-\alpha}\frac{w(1-\alpha)}{\alpha}}
    {1+\frac{\alpha}{1-\alpha}\frac{w(1-\alpha)}{\alpha}}
    =\frac{w}{w+1}
\end{align*}
The inequality holds because the likelihood ratio is greater than
$w(1-\alpha)/\alpha$ if $\hat{\mathbf{1}}^{Rx}=1$ and $x/(c+x)$ is a
increasing function with respect to $x$. Reversely, the conditional
probability $\Pr(\mathbf{1}^{link}=1|\hat{\mathbf{1}}^{link}=1)$ is
well-defined if and only if $\Pr(\hat{\mathbf{1}}^{link}=1)\neq0$,
which implies $\alpha\geq\alpha_C$.
\end{proof}

In \textbf{Lemma \ref{l7}},
$\Pr(\mathbf{1}^{link}=1|\hat{\mathbf{1}}^{link}=1)$ could be
interpreted as the probability of successful transmission in CRN and
the weighting factor in Bayesian risk (\ref{e30}) can be determined
by the constraint on the outage probability
$P_{out}=\Pr(\mathbf{1}^{link}=0|\hat{\mathbf{1}}^{link}=1)$. That
is, if a CRN maintains $P_{out}<\zeta$, $w=(1-\zeta)/\zeta$.
Therefore, the condition that allows CR-TX forwarding packets to
CR-Rx (i.e. $\alpha$ belongs to admissive region) guarantees the
outage probability of CR link. Further considering the proposed
physical layer models, we can establish and define a geographic
region, where CR-Tx is allowed forwarding packets to CR-Rx as long
as CR-Rx lies in the region.

\begin{defi}
\textbf{Neighborhood} of CR-Tx $\mathcal{N}$ is
$\{(r_{Rx},\theta_{Rx})|\alpha\geq\alpha_C\}$ or equivalently
becomes
$\{(r_{Rx},\theta_{Rx})|\Pr(\mathbf{1}^{link}=1|\hat{\mathbf{1}}^{link}=1)\geq
w/(w+1)\}$. \textbf{Coverage} of CR-Tx is neighborhood of CR-Tx
without PS.
\end{defi}

Please not that the coverage of CR-Tx is defined without the
existence of PS and the neighborhood is the effective area in real
operation coexisting with PS. When we consider a path loss model
between CR-TX and CR-Rx, coverage becomes a circularly shaped
region. However, due to hidden terminal problem as in Fig.
\ref{Fig_1} and \ref{Fig_2}, where PS is either apart from CR-Tx or
is blocked by obstacles, the probability of collision at CR-Rx could
increase and CR-Tx may be prohibited from forwarding packets to
CR-Rx. Therefore, neighborhood of CR-Tx shrinks from its coverage
and is no longer circular shape. In addition, hidden terminal
problem is location dependent, that is, PS is hidden to CR-Tx but
not to CR-Rx in Fig. \ref{Fig_1} and \ref{Fig_2}. Thus, CR-Rx is
possibly allowed forwarding packets to CR-Tx. From such
observations, CR links are directional and can be mathematically
characterized as follows.

\begin{defi}
$CR_i$ is said to be \textbf{connective} to $CR_j$ if $CR_j$ is
located in the neighborhood of $CR_i$, which is denoted by
$\mathbf{1}^{link}_{ij}=1$. Otherwise, $\mathbf{1}^{link}_{ij}=0$ if
$CR_i$ is not \textbf{connective} to $CR_j$.
\end{defi}

According to above arguments, it is possible that
CR-Rx is connective to CR-Tx but the reserve is not true. Mathematical
conclusion is developed in the following, and is numerically verified
in Fig. \ref{Fig_4}and \ref{Fig_5} in Section VI.

\begin{prop} \label{p6}
Connective relation is asymmetric, that is, for two cognitive
radios, $CR_i$ is connective to $CR_j$ does not imply $CR_j$ is
connective to $CR_i$, or mathematically,
$\mathbf{1}^{link}_{ij}=1\nRightarrow\mathbf{1}^{link}_{ji}=1$.
\end{prop}

\begin{proof}
We analytically illustrate using Fig. \ref{Fig_1}, where $CR_i$ lies
in the middle of $CR_j$ and PS-Tx and PS-Tx is hidden to $CR_j$ but
not to $CR_i$. Let $w=9$ to guarantee the outage probability of CR
link less than $0.1$ and let $\Pr(\mathbf{1}^{PS}=1)=0.7$, i.e., the
spectrum utility of PS is only $30\%$. If $CR_i$ wants to forward
packets to $CR_j$ (i.e. $CR_i$ is CR-Tx and $CR_j$ is CR-Rx), $CR_i$
can successfully detect the activity of PS and
$\Pr(\mathbf{1}^{Tx}=1|\mathbf{1}^{PS}=1)\approx1$ and
$\Pr(\mathbf{1}^{Tx}=1|\mathbf{1}^{PS}=0)\approx0$. Therefore,
$CR_i$ forwards packets to $CR_j$ only when $\mathbf{1}^{PS}=1$. In
addition, since $CR_i$ is located in the transmission range of
$CR_j$, $CR_j$ is located in the transmission range of $CR_i$ in a
pure path loss model and
$\Pr(\mathbf{1}^{Rx}=1|\mathbf{1}^{Tx}=1,\mathbf{1}^{PS}=1)\approx1$.
Applying (\ref{e38}) and (\ref{e13}), we have $\alpha\approx1$ and
$\mathbf{1}^{link}_{ij}=1$. On the other hand, when $CR_j$ wants to
forward packets to $CR_i$, $CR_j$ becomes CR-Tx and $CR_i$ becomes
CR-Rx. Since PS-Tx is hidden to $CR_j$, at $CR_j$, the received
signal power from PS is below noise power, and $\mu_{CR}=\mu_0$ and
$\sigma_{CR}^2=\sigma_0^2$ in (\ref{e42}) and (\ref{e43}).
Therefore, $\Pr(\mathbf{1}^{Tx}=1|\mathbf{1}^{PS}=1)\approx1$ and
$\Pr(\mathbf{1}^{Tx}=1|\mathbf{1}^{PS}=0)\approx1$. That is, $CR_j$
always feels the spectrum available and intends to forward packets
to $CR_i$. However, when $\mathbf{1}^{PS}=0$, collisions occurs at
$CR_i$ and
$\Pr(\mathbf{1}^{Rx}=1|\mathbf{1}^{Tx}=1,\mathbf{1}^{PS}=0)\approx0$.
Similarly, by (\ref{e38}) and (\ref{e13}), we have
$\alpha\approx\Pr(\mathbf{1}^{PS}=1)=0.7<w/(w+1)=0.9$ and
$\mathbf{1}^{link}_{ji}=0$.
\end{proof}

\textbf{Proposition \ref{p6}} mathematically suggests that links in
CRN are generally asymmetric and even unidirectional as the argument
in \cite{GSSinCRN:Centin09}. Therefore, traditional feedback
mechanism such as acknowledgement and automatic repeat request (ARQ)
in data link layer may not be supported in general. This challenge
can be alleviated via cooperative schemes. Roles of a cooperative
node in CR network operation thus include
\begin{enumerate}
    \item Extend neighborhood of CR-Tx to its coverage
    \item Ensure bidirectional links in CRN (i.e. enhance probability to maintain bidirectional)
    \item Enable feedback mechanism for the purpose of upper layers
\end{enumerate}
Since neighborhood increases as $\alpha_C$ decreases, by
\textbf{Lemma \ref{l6}}, the capability of cooperative schemes to
extend neighborhood increases when number of cooperative nodes
increases. Therefore, spectrum sensing capability mathematically
determine CRN topology. It also suggests the functionality of
cooperative nodes in topology control
\cite{GSSinCRN:Thomas07}\cite{GSSinCRN:Chen072} and network routing
\cite{GSSinCRN:Centin09}, which is critical in CRN due to asymmetric
links and heterogeneous network architecture
\cite{GSSinCRN:Centin09}.

Here, we illustrate impacts of correlation between spectrum
availability at cooperative nodes on neighborhood. Recall
\textbf{Proposition \ref{p3}}, where we considered two cooperative
nodes with $\beta_i=\gamma_i$, $\rho_i>0, i=0,1$, and
$\mathbf{1}^{Co}_1\rhd\mathbf{1}^{Co}_2$. From (\ref{e37}), we have
\begin{equation} \label{e39}
    \alpha_C=
    \begin{cases}
    \alpha_4^C &\text{if $\Delta<(2\beta_2-1)(1-\beta_1)$}\\
    \alpha_2^C &\text{otherwise}
    \end{cases}
\end{equation}
Therefore, as $\Delta$ increases from $\Delta_{min}$,
$\alpha_C=\alpha_4^C$ increases from 0 and achieves maximum at
$\Delta=(2\beta_2-1)(1-\beta_1)$. At this point,
$\alpha_C=\alpha_4^C=\alpha_2^C=w(1-\gamma_1)/(\beta_1+w(1-\gamma_1))$,
which is the critical boundary with node one alone. If $\Delta$
further increases to $\Delta_{max}$, $\alpha_C=\alpha_2^C$ decreases
to 0. We conclude that positive correlation between
$\mathbf{1}^{Co}_1$ and $\mathbf{1}^{Co}_2$ shrinks the
neighborhood, compared to the independent case ($\rho_{12}=0$),
unless the correlation is high enough, i.e.,
$\Delta>(2\beta_2-1)(1-\beta_1)/\beta_2$ by solving
$\alpha_2^C|_{\Delta}<\alpha_4^C|_{\Delta=0}$ according to
(\ref{e39}).

If one CR has larger neighborhood area, it is expected to be
connective to more CRs and to have higher probability to forward
packets successfully and higher throughput of CRN accordingly. The
result offers a novel dimension to evaluate cooperative nodes. That
is, different from criterions in link level, such as minimum
Bayesian risk or maximum reliability as we mentioned in last
section, maximum neighborhood area is a novel criterion to select
the best cooperative node from the viewpoint of network operation.

\begin{prop} \label{p5}
(\textbf{Optimum Selection of Cooperative Node}) For a CRN with a
constraint on the outage probability $P_{out}<\zeta$, there are one
CR and $K$ cooperative nodes, indexed by $k$. The best cooperative
node for the CR under maximum neighborhood area criterion is
\begin{equation} \label{e31}
    k_{opt}=\arg\underset{k}{\max}\mathcal{N}_A(k) \quad \text{with $w=(1-\zeta)/\zeta$}
\end{equation}
where $\mathcal{N}_A(k)$ represents neighborhood area of the CR with
the aid of the $k$th cooperative node.
\end{prop}

In CRN, CRs could act as relay nodes to relay packets to the
destination. Assume the destination is in the direction $\theta$ of
a CR with respect to PS. It is intuitive for the CR to forward
packets to the direction around $\theta$, say $\theta\pm\epsilon$.
Let
$\mathcal{N}_{\theta\pm\epsilon}=\{\mathcal{N}|\theta_{Rx}\in(\theta-\epsilon,\theta+\epsilon)\}$
and then the best cooperative node may become (\ref{e31}) with
$\mathcal{N}$ replaced by $\mathcal{N}_{\theta\pm\epsilon}$.

\section{Experiments}
\subsection{General Spectrum Sensing}
\subsubsection{Spectrum Sensing Performance}
The performance of spectrum sensing, measured by Bayesian risk
(\ref{e30}), is plotted by Bayesian risk versus the probability of
spectrum availability at CR-Rx $\alpha$ in Fig. \ref{Fig_3}. We set
the weighting factor $w=9$ ($w$ is defined in (\ref{e41})) to
guarantee the outage probability of CR link less than $0.1$. Larger
Bayerian risk represents worse performance because spectrum sensing
induces more possibility of collisions at CR-Rx or of losing
opportunity to utilize spectrum. We see that traditional spectrum
sensing without considering spectrum availability at CR-Rx (i.e.
$\mathbf{1}^{Rx}$) has large Bayesian risk when $\alpha$ becomes
small because collisions usually occur when CR-Tx determines link
availability only by localized spectrum availability at CR-Tx (i.e.
$\mathbf{1}^{Tx}$). On the other hand, by considering
$\mathbf{1}^{Rx}$ in our general spectrum sensing (\ref{e13}) with
known $\alpha$, Bayesian risk decreases when $\alpha$ is less than
$0.9$, which is the critical boundary of $\alpha$, $\alpha_C$, and
CR-Tx is prohibited from forwarding packets to CR-Rx. Therefore,
risk occurs due to loss of opportunity to utilize spectrum.

However, in practice, $\alpha$ is unknown and needs to be estimated
by \textbf{Lemma \ref{l8}}. We set observation depth (i.e. duration)
$L=15$ and show expected Bayesian risk of inference-based spectrum
sensing (\ref{e13}) with respect to observed sequence
$\mathbf{1}^{Rx}[n-1],\ldots,\mathbf{1}^{Rx}[n-L]$. The performance
degrades around $\alpha_C$ and even worse than that of traditional
spectrum sensing when $\alpha\geq\alpha_C$ because the estimation
error may cause the estimated $\alpha$ (\ref{e1}) to across
$\alpha_C$ and results in different sensing rules; however, it is
close to the performance with known $\alpha$. This verifies the
effectiveness of inference-based spectrum sensing.

Fig. \ref{Fig_3} also shows Bayesian risk of cooperative sensing
(\ref{e15}) under different sensing capability of a cooperative
node, i.e. reliability $M_R$. We assume statistical information
$\{\alpha,\beta,\gamma\}$ can be perfectly estimated. The
performance curve is composed of three line segments as in
(\ref{e15}) and shows performance improvement in the middle segment
due to the aid of the cooperative node. However, in the right and
the left segments, the cooperative node becomes useless and the
performance is equal to that of non-cooperative sensing. In
comparison of sensing capability of a cooperative node, the one with
larger reliability is expected to have higher correlation between
spectrum availability at CR-Rx and the cooperative node (i.e.
$\mathbf{1}^{Co}$) and to provide more information about
$\mathbf{1}^{Rx}$; therefore, it achieves lower Bayesian risk and
lower $\alpha_C$ (i.e. larger admissive region). In addition, when
$\beta+\gamma=1$ ($\beta=0.8,\gamma=0.2$), $\mathbf{1}^{Rx}$ and
$\mathbf{1}^{Co}$ are independent and the performance is identical
to that of spectrum sensing with known $\alpha$ in Fig. \ref{Fig_3}.

\subsubsection{Impacts of Correlation between Spectrum Availability at Cooperative Nodes}
We next investigate performance of spectrum sensing with two
cooperative nodes with respect to the correlation between spectrum
availability at these two nodes (i.e. $\mathbf{1}^{Co}_1$ and
$\mathbf{1}^{Co}_2$) $\rho_{12}$. We set $\beta_1=\gamma_1=0.75$,
$\beta_2=\gamma_2=0.7$ as the scenario in \textbf{Proposition
\ref{p3}} and depict Bayesian risk in Fig. \ref{Fig_6}. Generally
speaking, Bayesian risk decreases (increases) as $\rho_{12}$
increases in high (low) $\alpha$. We also observe that $\alpha_C$
decreases when number of cooperative nodes increases and $\alpha_C$
increases when $\rho_{12}$ increases unless the correlation is high
enough. For example, if two cooperative nodes are close in location
and $\rho_{12}=0.8$, $\alpha_C$ is less than that when
$\rho_{12}=0$. It is also interesting to note that there exists a
region of $\alpha$ such that the identity of $\mathbf{1}^{Co}_1$ and
$\mathbf{1}^{Co}_2$ instead of $\mathbf{1}^{Co}_1$ determines CR
link availability as in (\ref{e40}) and Bayesian risk is less than
that with $\mathbf{1}^{Co}_1$ alone.

The results further suggest trade-off between performance in link
layer and network layer when we select cooperative nodes. That is,
for one CR link with $\alpha>0.5$ (e.g. CR-Rx is close to CR-Tx or
spectrum utilization of PS is low), large $\rho_{12}$ is preferred
to achieve low risk. However, from network perspective, to achieve
high number of CR links that are admissive to CR-Tx (i.e. to achieve
large neighborhood and low $\alpha_C$) and thus high throughput of
CRN, small $\rho_{12}$ is preferred.

\subsubsection{Robust Spectrum Sensing}
For multiple cooperative nodes, with six nodes in our simulation, we
show the performance of cooperative spectrum sensing with limited
statistical information $\mathbf{Q}^{(s)}_{k,K},s=0,1$ due to
limited sensing duration. We first find least-favorable joint pmf
$\mathbf{P}^{(s)}_{opt},s=0,1$ by (\ref{e35}) and then compute
corresponding Bayesian risk, which is shown in Fig. \ref{Fig_7}
under different order of known marginal pmf $k$ (i.e. capability of
observation). That is, CR-Tx only acquires pmf of $k$ out of six
cooperative nodes. The risk is compared to that with the optimum
sensing rule (\ref{e14}) and that with assumption of independent
spectrum availability (\ref{e34}). Obviously, Bayesian risk
decreases and approaches to that in the optimum case as the order of
known marginal pmf $k$ increases because more information is
acquired to generate $\mathbf{P}^{(s)}_{opt},s=0,1$ closer to the
true one $\mathbf{P}^{(s)}_{K},s=0,1$. We observe that when the
order $k$ is greater than 3, robust spectrum sensing outperforms the
case of traditional independence assumption. Therefore, if CR-Tx
would like to select six cooperative nodes, CR-Tx only requires
statistical information about spectrum availability among three out
of six cooperative nodes, i.e. $\mathbf{Q}^{(s)}_{3,6},s=0,1$, to
achieve better performance than the case according to reliability
criterion.

\subsection{Neighborhood of CR-Tx}
\subsubsection{Without Obstacles}
In Fig. \ref{Fig_4}, we illustrate neighborhood of CR-Tx ("$+$" in
the figure) without blocking. The neighborhood boundary with/without
a cooperative node ("$\circ$" in the figure) is depicted by a thick
and a thin line respectively. The parameters are set as follows:
$\mu_0=0$, $\sigma_0^2=1$, $\sigma_S^2=8$, $K_0=10$, $a=3$, $L_0=3$,
$\tau_{Tx}=\tau_{Co}=3$, and $\Pr(\mathbf{1}^{PS}=1)=0.6$.

In Fig. \ref{Fig_4}(a), PS ("$\ast$" in the figure) is placed near
to CR-Tx ($(0.7,0)$). We observe that CR-Tx almost perfectly detects
the state of PS, i.e.,
$\Pr(\mathbf{1}^{Tx}=1|\mathbf{1}^{PS}=1)\approx1$ and
$\Pr(\mathbf{1}^{Tx}=1|\mathbf{1}^{PS}=0)\approx0$, and
$\alpha\approx\Pr(\mathbf{1}^{Rx}=1|\mathbf{1}^{Tx}=1,\mathbf{1}^{PS}=1)$
by (\ref{e38}). Therefore, neighborhood of CR-Tx approaches to its
coverage and the cooperative node is not necessary in this case from
a viewpoint of network operation. However, when PS is apart from
CR-Tx ($(1.7,0)$) as we shown in Fig. \ref{Fig_1}, the neighborhood
at PS side shrinks and is no longer circularly shaped because PS is
hidden to CR-Tx and hence probability of collision at CR-Rx
increases when $\mathbf{1}^{PS}=0$. Fig. \ref{Fig_4}(b)$\sim$(d)
illustrate the neighborhood under different locations of the
cooperative node. We observe that neighborhood area decreases when
the cooperative node moves away from PS and there even exists a
region where cooperative sensing can not help. Therefore, the
cooperative node in Fig. \ref{Fig_4}(b) is the best among these
three nodes according to maximum neighborhood area criterion in
\textbf{Proposition \ref{p5}}.

We present an example of existence of unidirectional link in CRN. In
Fig. \ref{Fig_4}(b), assume one CR is located at $(1,0)$. Obviously,
the CR-Tx is not connective to the CR and therefore is prohibited
from forwarding packets to the CR. However, by Fig. \ref{Fig_4}(a),
the CR is connective to CR-Tx, which makes the link unidirectional
(only from the CR to CR-Tx). As \textbf{Proposition \ref{p6}}, this
also shows asymmetric connective relation even under rather ideal
radio propagation. With the aid of a cooperative node located at
$(0.4,0.3)$, the link returns to a bidirectional link.

\subsubsection{With Obstacles}
Alternatively, we consider effects of shadowing due to blocking, as
we illustrated in Fig. \ref{Fig_2}. We set shadowing parameters
$b_{Tx}=25$ and parameter of obstacle size $\kappa=0.3$ and $0.7$ in
Fig. \ref{Fig_5}(a)(b) and Fig. \ref{Fig_5}(c)(d) respectively. We
observe that small obstacles size (i.e. small $\kappa$) can result
in more substantial shrink of the neighborhood, compared to large
obstacles size (i.e. large $\kappa$). The reason is: if $\kappa$ is
small, only a small region around CR-Tx falls in deep shadowing and
the state of PS can be successfully detected outside that region.
Therefore, this leads to high probability of collision at CR-Rx as
$\mathbf{1}^{PS}=0$. On the other hand, if $\kappa$ is large, CRs
are likely separated from PS by obstacles, which results in large
"distance" between CR and PS. Here, "distance" is measured by
received signal power \cite{GSSinCRN:Tu09}\cite{GSSinCRN:Chen07}. In
comparison of the capability of a cooperative node, the one in small
$\kappa$ has good capability to recover the neighborhood to its
coverage even when the node is at opposite side of PS. However, for
large $\kappa$, the cooperative node may also be in deep shadowing
and becomes useless to recover neighborhood of CR-Tx.

\section{Conclusion}
In this paper, we showed that CR link availability should be
determined by spectrum availability at both CR-Tx and CR-Rx, which
may not be identical due to hidden terminal problem (Fig.
\ref{Fig_1} and \ref{Fig_2}). In order to fundamentally explore the
spectrum sensing at link level and its impacts on network operation,
we established an indicator model of CR link availability and
applied statistical inference to predict/estimate unknown spectrum
availability at CR-Tx due to no centralized coordinator nor
information exchange between CR-Tx and CR-Rx in advance. We
therefore expressed conditions for CR-Tx to forward packets to CR-Rx
under guaranteed outage probability. These conditions, along with
physical channel models, define neighborhood of CR-Tx, which is no
longer circularly shaped as coverage. This results in asymmetric or
even unidirectional links in CRN, as we illustrated in Section VI.
The impairment of CR links can be alleviated via cooperative scheme.
Therefore, spectrum sensing capability determines network topology
and thus throughput of CRN. Several factors with impacts on spectrum
sensing are analyzed, including:
\begin{enumerate}
\item Correlation of spectrum availability at cooperative nodes
\item Capability of observation at CR-Tx (i.e. available statistical information at CR-Rx)
\item Locations of cooperative nodes and environment (i.e. obstacles)
\end{enumerate}

Furthermore, limits of cooperative scheme were also addressed in
link level and network level. In addition, to measure sensing
capability and then to select cooperative nodes is an important
issue because we would like to minimize information exchange to
increase spectrum utilization. Criterions from link level (maximum
reliability or minimum Bayesian risk) and network level (maximum
neighborhood area) perspectives were accordingly proposed. We
numerically demonstrated existence of trade-off in designing systems
in different layers. In addition, robust spectrum sensing was
proposed to deal with local and partial information due to no
centralized coordination and limited sensing duration in CRN. More
useful results in CRN extended from this research can be expected in
future works.


\begin{figure} [!b]
\begin{center}
\includegraphics[height=20em]{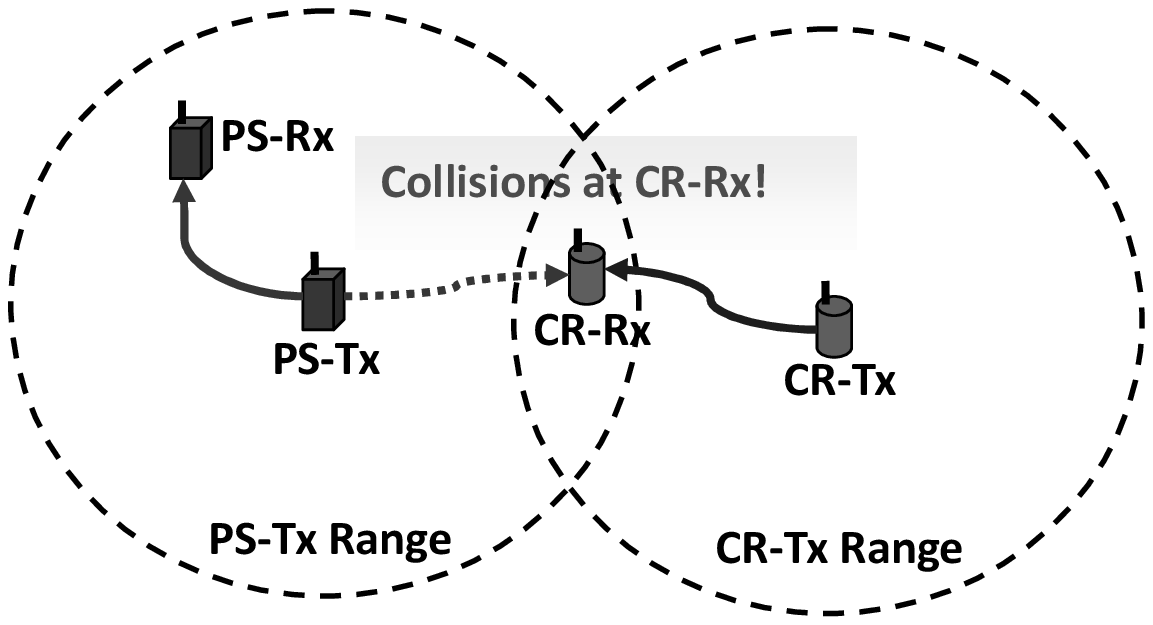}
\end{center}
\caption{Hidden terminal problem. CR-Rx lies in the middle of CR-Tx
and PS-Tx and PS-Tx is hidden to CR-Tx.} \label{Fig_1}
\end{figure}

\begin{figure} [!b]
\begin{center}
\includegraphics[height=25em]{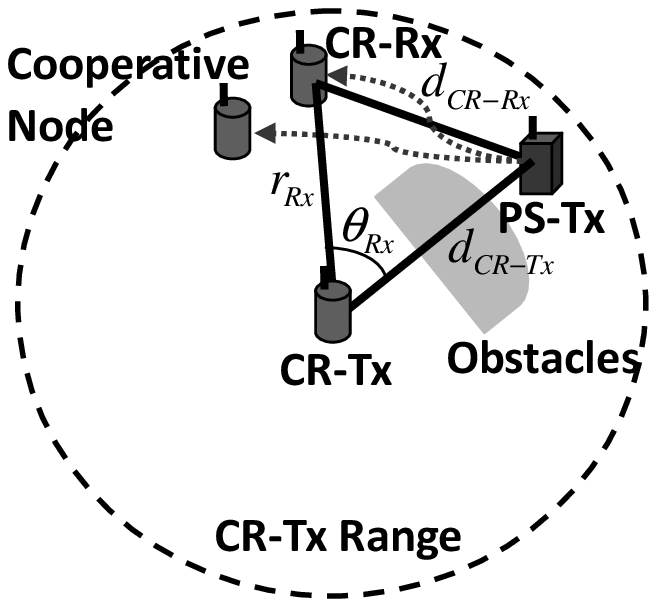}
\end{center}
\caption{Activities of PS-Tx is blocked by obstacles to CR-Tx but
not to CR-Rx.} \label{Fig_2}
\end{figure}

\begin{figure} [!b]
\begin{center}
\includegraphics[height=25em]{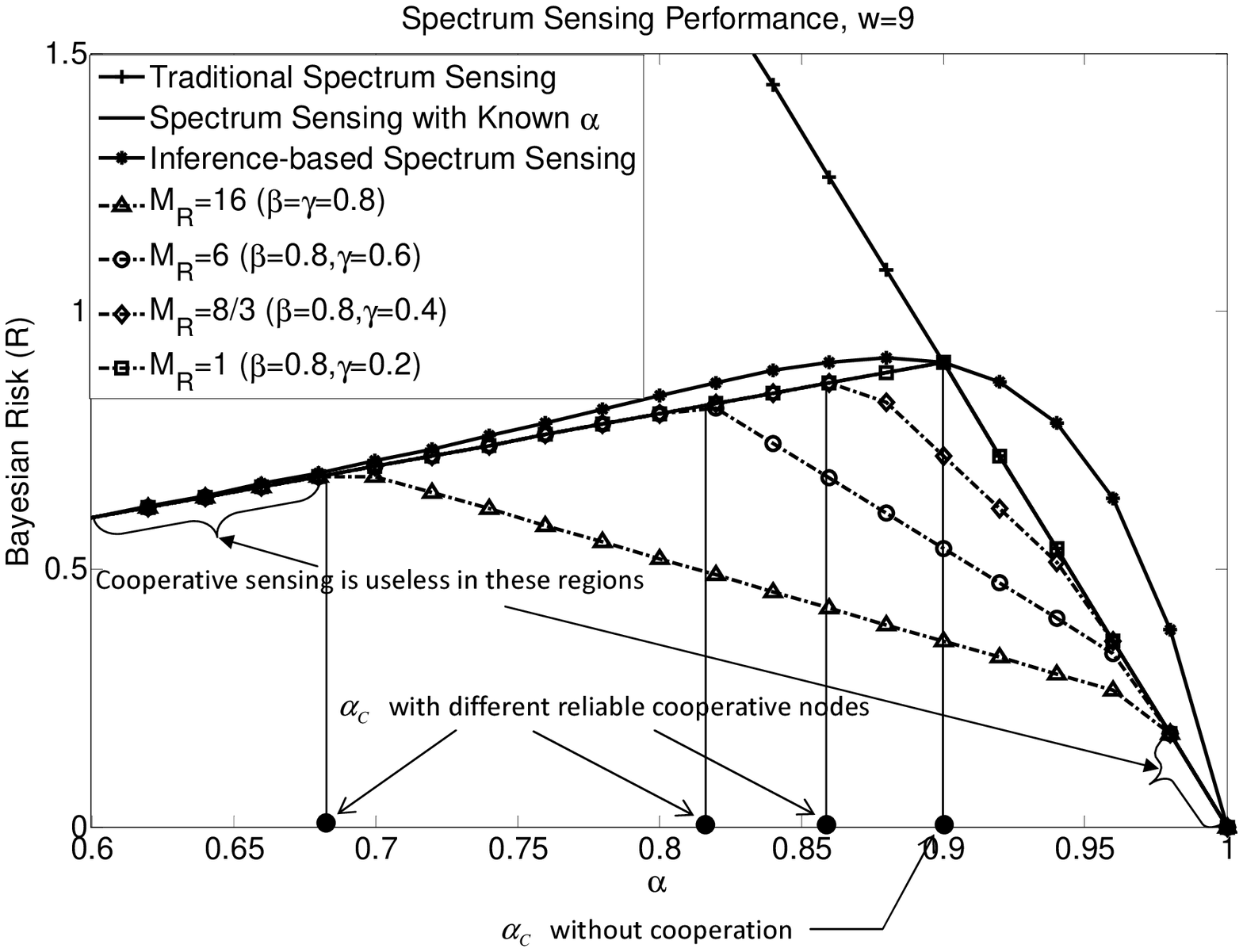}
\end{center}
\caption{Performance (Bayesian risk) of spectrum sensing schemes
with respect to the probability of spectrum availability at CR-Rx
$\alpha$. The critical boundary of $\alpha$, $\alpha_C$, is
represented by $\bullet$ under different sensing schemes .}
\label{Fig_3}
\end{figure}

\begin{figure} [!b]
\begin{center}
\includegraphics[height=25em]{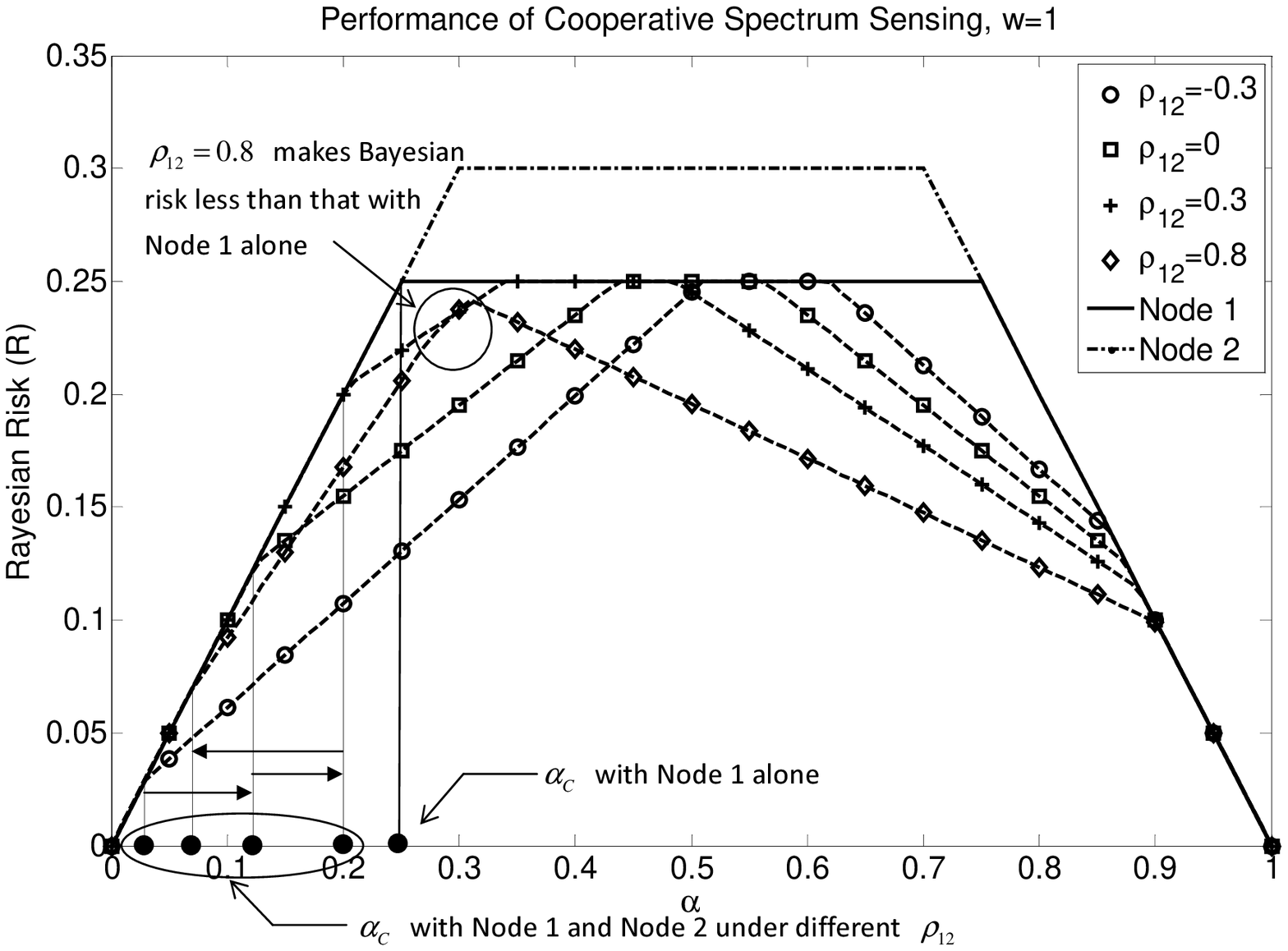}
\end{center}
\caption{Performance of spectrum sensing with one/two cooperative
node(s) with respect to the probability of spectrum availability at
CR-Rx $\alpha$. The transition and the value of the critical
boundary of $\alpha$, $\alpha_C$, under different correlation of
Node 1 and Node 2 $\rho_{12}$ is also shown.} \label{Fig_6}
\end{figure}

\begin{figure} [!b]
\begin{center}
\includegraphics[height=25em]{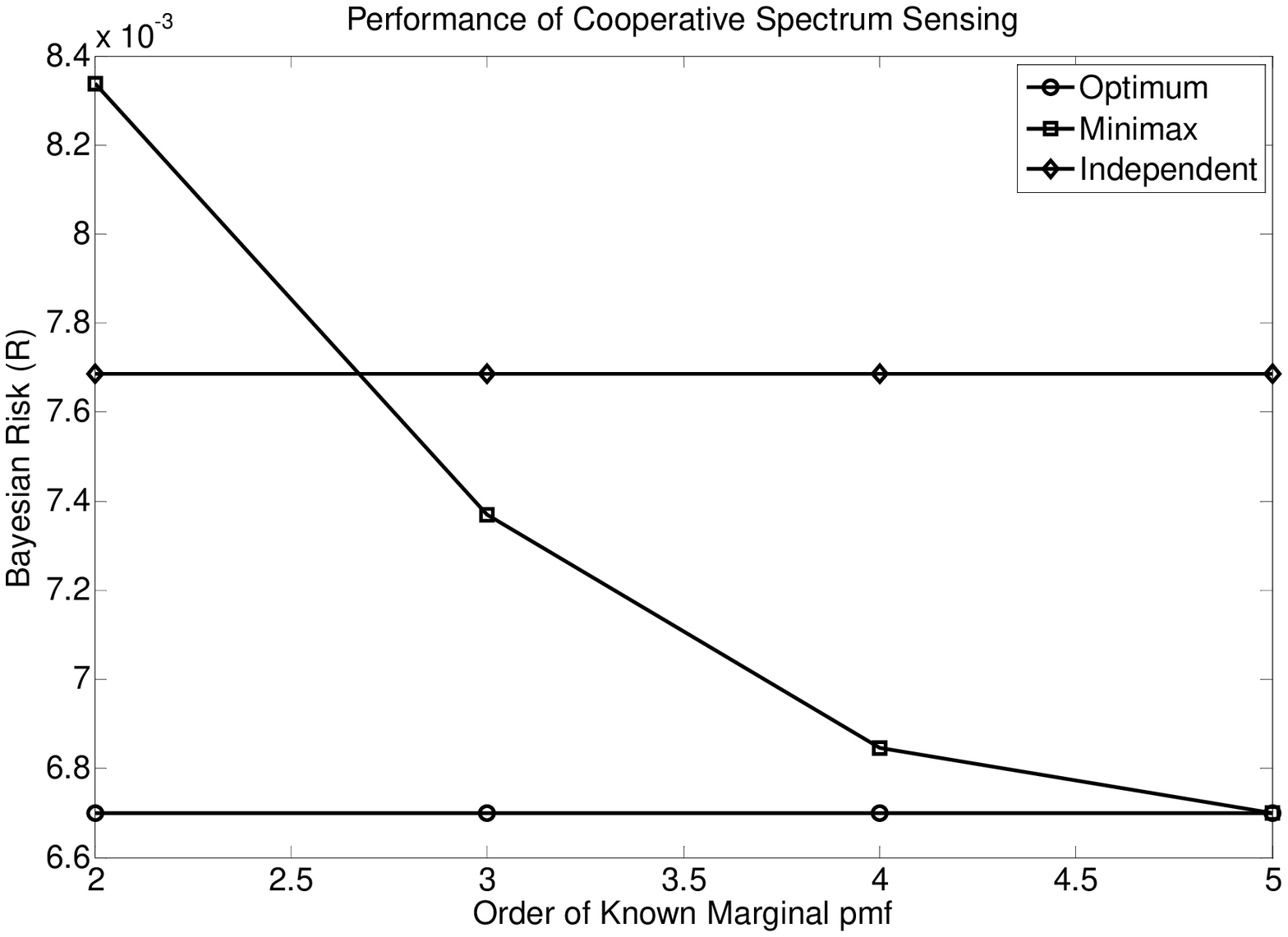}
\end{center}
\caption{Performance of robust spectrum sensing under different
order of known marginal pmf $k$, that is, pmf of $k$ out of six
cooperative nodes is available at CR-Tx. The performance is compared
to the spectrum sensing with full statistical information and with
traditional independence assumption.} \label{Fig_7}
\end{figure}

\begin{figure} [!b]
\begin{center}
\includegraphics[height=35em]{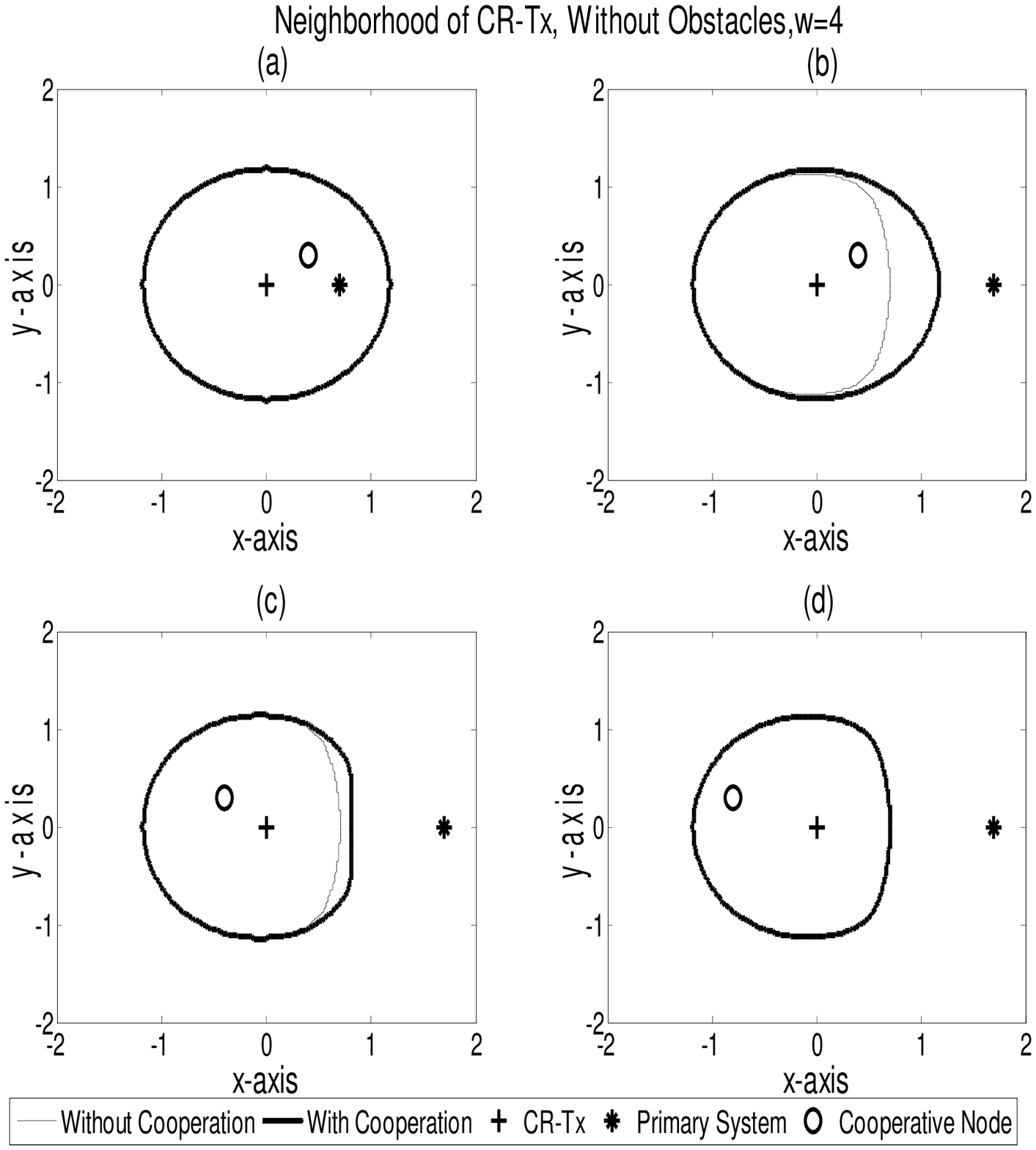}
\end{center}
\caption{Neighborhood of CR-Tx without obstacles. CR-Tx can only be
allowed forwarding to CR-Rx located within the bounded region.}
\label{Fig_4}
\end{figure}

\begin{figure} [!b]
\begin{center}
\includegraphics[height=35em]{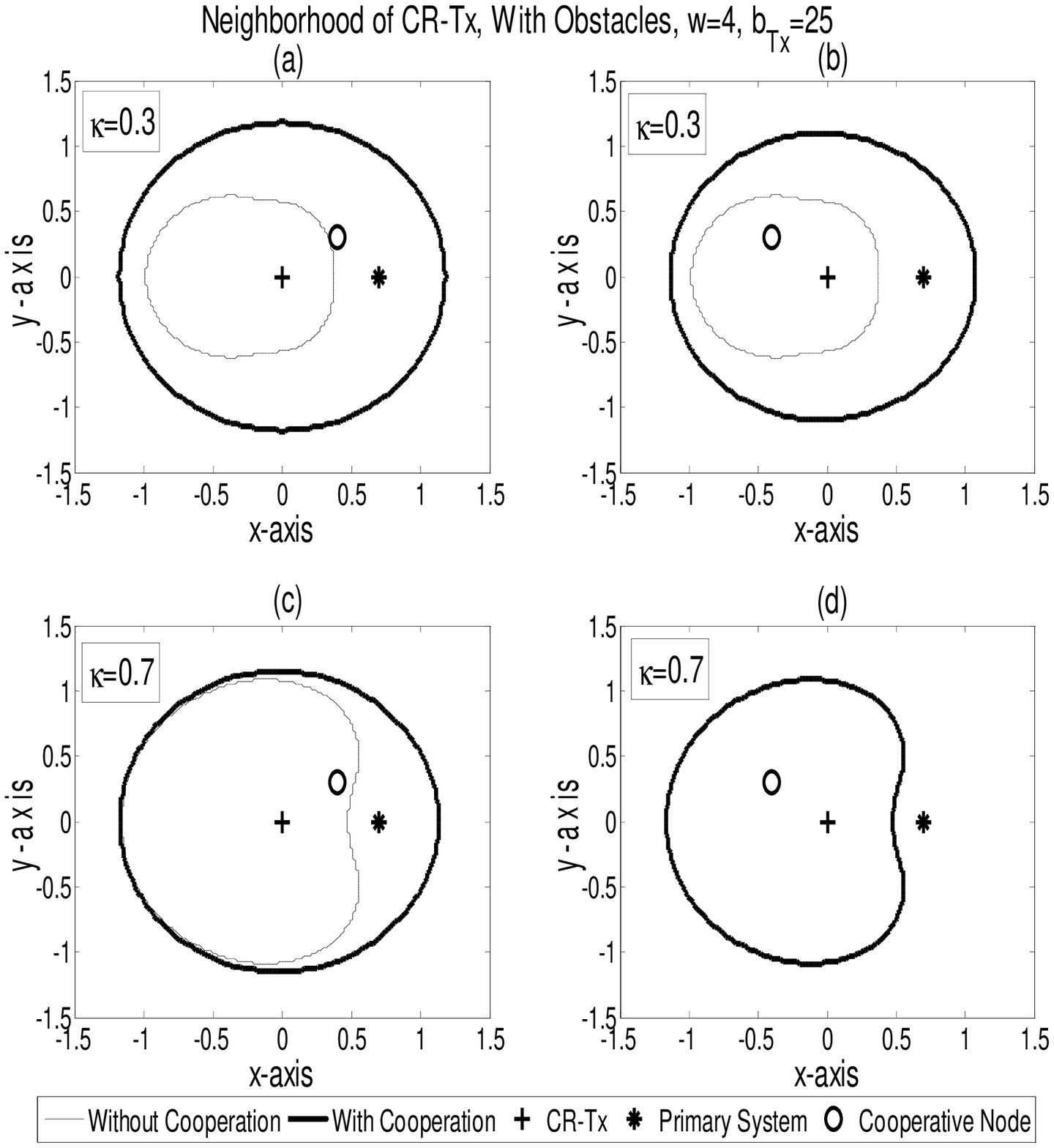}
\end{center}
\caption{Neighborhood of CR-Tx with obstacles. Effects of small
($\kappa=0.3$ in (a)(b)) and large ($\kappa=0.7$ in (c)(d))
obstacles are compared.} \label{Fig_5}
\end{figure}

%

%

\end{document}